\newif\ifTR
\TRtrue

\documentclass[a4paper,UKenglish,cleveref,autoref,thm-restate]{lipics-v2021}

\title{Complementing Emerson-Lei Elevator Automata \ifTR (Technical Report)\fi}

\author{Ondrej {Alexaj}}{Brno University of Technology, Czech Republic}{xalexa09@stud.fit.vutbr.cz}{https://orcid.org/0009-0009-3994-4563}{}
\author{Vojt\v{e}ch {Havlena}}{Brno University of Technology, Czech Republic}{ihavlena@fit.vutbr.cz}{https://orcid.org/0000-0003-4375-7954}{}
\author{Ond\v{r}ej {Leng\'al}}{Brno University of Technology, Czech Republic}{lengal@fit.vutbr.cz}{https://orcid.org/0000-0002-3038-5875}{}
\author{Yong {Li}}{Key Laboratory of System Software (Chinese Academy of Sciences), ISCAS, China}{liyong@ios.ac.cn}{https://orcid.org/0000-0002-7301-9234}{}
\author{Nicolas {Mazzocchi}}{Slovak University of Technology in Bratislava, Slovak Republic}{nicolas.mazzocchi@stuba.sk}{https://orcid.org/0000-0001-6425-5369}{}

\authorrunning{O.~Alexaj, V.~Havlena, O.~Leng\'al, Y.~Li, and N.~Mazzocchi} 

\Copyright{Ondrej {Alexaj}, Vojt\v{e}ch {Havlena}, Ond\v{r}ej {Leng\'al}, Yong {Li}, Nicolas {Mazzocchi}} 



\category{}

\ifTR
\else
\relatedversiondetails[cite={techrep}]{Technical Report}{http://arxiv.org/abs/2606.26768}
\fi

\supplementdetails[subcategory={Source Code}]{Software}{https://github.com/VeriFIT/kofola}

\funding{
This work was supported in part by
the CAS Project for Young Scientists in Basic Research (Grant No.\ YSBR-040), the National Natural Science Foundation of China (Grant No.\ 62102407), the Beijing Natural Science Foundation (Grant No.\ IS26039),
the Czech Science Foundation project \mbox{26-22640S}, and
the FIT BUT internal project FIT-S-26-9011.
}

\nolinenumbers 

\ifTR
  \hideLIPIcs
\else
  \EventEditors{Ana Sokolova and Patrick Totzke}
  \EventNoEds{2}
  \EventLongTitle{37th International Conference on Concurrency Theory (CONCUR 2026)}
  \EventShortTitle{CONCUR 2026}
  \EventAcronym{CONCUR}
  \EventYear{2026}
  \EventDate{September 1--4, 2026}
  \EventLocation{Liverpool, UK}
  \EventLogo{}
  \SeriesVolume{391}
  \ArticleNo{27}
\fi

\usepackage[table,dvipsnames]{xcolor}
\usepackage{xspace}
\usepackage{amsmath}
\usepackage{amssymb}
\usepackage{booktabs}
\usepackage{subcaption}
\usepackage{tikz}
\usepackage{mathtools}
\usepackage[noend,ruled,vlined,linesnumbered]{algorithm2e}

\usepackage{algpseudocode}

\usetikzlibrary{shapes.geometric}
\usetikzlibrary{backgrounds}
\usetikzlibrary{fit}

\usetikzlibrary{automata}
\usetikzlibrary{arrows.meta}
\usetikzlibrary{bending}
\usetikzlibrary{shapes.callouts}
\usetikzlibrary{quotes}
\usetikzlibrary{positioning}
\usetikzlibrary{calc}
\usetikzlibrary{matrix}

\definecolor{spotblue}{RGB}{31,120,180}
\definecolor{spotpink}{RGB}{255,77,160}
\definecolor{spotorange}{RGB}{255,127,0}
\definecolor{spotpurple}{RGB}{106,61,154}
\definecolor{spotgreen}{RGB}{51,160,44}
\definecolor{spotred}{RGB}{227,26,28}
\definecolor{spotyellowish}{RGB}{196,196,0}
\definecolor{spotgray}{RGB}{80,80,80}
\definecolor{spotlight blue}{RGB}{107,246,255}
\definecolor{spotlight pink}{RGB}{255,154,255}
\definecolor{spotlight orange}{RGB}{255,156,103}
\definecolor{spotlight purple}{RGB}{178,164,255}
\definecolor{spotlight green}{RGB}{167,237,121}
\definecolor{spotlight red}{RGB}{255,104,104}
\definecolor{spotlight yellowish}{RGB}{255,224,64}
\definecolor{spotlight gray}{RGB}{192,192,144}

\tikzset{
  >={Stealth[round,bend]},
}

\tikzstyle{automaton}=[
  semithick,shorten >=0pt,>={Stealth[round,bend]},
  node distance=1.5cm,
  initial text=,
  every initial by arrow/.style={every node/.style={inner sep=0pt}},
  every state/.style={
    align=center,
    fill=white,
    minimum size=7.5mm,
    inner sep=0pt,
    execute at begin node=\strut,
  },%
]
\tikzset{
  scc/.style={draw=gray,fill=black!10,rounded corners=2mm},
  lstate/.style={state},
  }

\tikzstyle{smallautomaton}=[
  automaton,
  node distance=7mm,
  every state/.style={minimum size=4mm,fill=white,inner sep=1.5pt}
]

\tikzstyle{mediumautomaton}=[
  automaton,
  node distance=1cm,
  every state/.style={minimum size=6mm,fill=white,inner sep=2pt}
]

\tikzstyle{cstate}=[state,capsule,text width=,inner xsep=-5pt]
\tikzstyle{dot}=[fill=black,circle,minimum size=4pt,inner sep=0]

\makeatletter
\tikzoption{initial angle}{\tikzaddafternodepathoption{\def\tikz@initial@angle{#1}}}
\makeatother

\tikzstyle{initial overlay}=[every initial by arrow/.append style={overlay}]

\tikzstyle{unreachable} = [densely dotted]

\tikzstyle{acclabel} = [
  fill=yellow!20,
  rounded corners=3pt,
  draw=black,
  inner ysep=0pt,
  inner xsep=3pt,
]
\tikzstyle{ltllabel} = [acclabel,fill=darkgreen!20]

\tikzstyle{namelabel} = [
  execute at begin node = {$($},%
  execute at end node = {$)$}
]

\tikzstyle{matrix of states} = [
matrix of nodes,
every node/.style={state},
column sep=.8cm,
row sep=.8cm,
execute at empty cell={
  \node[transparent,name=
    \tikzmatrixname-\the\pgfmatrixcurrentrow-\the\pgfmatrixcurrentcolumn]{};]},
]

\tikzstyle{accset}=[
  rectangle,inner sep=1.5pt,     
	rounded corners=1.5mm,
  draw=none, 
  solid,
  collacc0, text=white,
  thin,
  anchor=center,
  minimum size=3.3mm,
  text width={}, 
  font=\bfseries\sffamily\footnotesize
]
\tikzstyle{accsquare}=[accset,rectangle,inner sep=1.9pt,rounded corners=0pt]

\tikzset{
  collacc0/.style={fill=spotblue},
  collacc1/.style={fill=spotpink},
  collacc2/.style={fill=spotorange},
  collacc3/.style={fill=spotpurple},
  collacc4/.style={fill=spotgreen},
  collacc5/.style={fill=spotred},
  collacc6/.style={fill=spotyellowish,draw=black,text=black},
  collacc7/.style={fill=spotgray},
  collacc8/.style={fill=spotlight blue,draw=black,text=black},
  collacc9/.style={fill=spotlight pink},
  collacc10/.style={fill=spotlight orange},
  collacc11/.style={fill=spotlight purple},
  collacc12/.style={fill=spotlight green},
  collacc13/.style={fill=spotlight red},
  collacc14/.style={fill=spotlight yellowish},
  collacc15/.style={fill=spotlight gray},
}

\pgfkeyssetvalue{/sacc/where}{center}
\tikzset{
  sacc where/.code={
    \pgfkeyssetvalue{/sacc/where}{#1}
  }
}

\tikzset{
  l/.pic={\node[outer sep=2pt] {#1};},%
  acc/.pic={\node[accset,collacc#1]{\upshape #1};},%
  accsq/.pic={\node[accsquare,collacc#1]{#1};},%
  eacc/.pic={\node[accset,collacc#1]{\emptyacc};},%
  eaccsq/.pic={\node[accsquare,collacc#1]{\emptyacc};},%
  sacc/.style = {
    append after command=
      {pic at (\tikzlastnode.\pgfkeysvalueof{/sacc/where}) {acc=#1}}
  },
  saccsq/.style = {
    append after command=
      {pic at (\tikzlastnode.\pgfkeysvalueof{/sacc/where}) {accsq=#1}}
  },
  esacc/.style = {
    append after command=
      {pic at (\tikzlastnode.\pgfkeysvalueof{/sacc/where}) {eacc=#1}}
  },
  esaccsq/.style = {
    append after command=
      {pic at (\tikzlastnode.\pgfkeysvalueof{/sacc/where}) {eaccsq=#1}}
  },
  pics/cacc/.style 2 args={%
    code={\node[accset,collacc#1]{#2};}%
  },%
  pics/caccsq/.style 2 args={%
      code={\node[accsquare,collacc#1]{#2};}%
    },%
  csacc/.style 2 args = {%
      append after command=%
        {pic at (\tikzlastnode.\pgfkeysvalueof{/sacc/where}) {cacc={#1}{#2}}}%
  },%
  csaccsq/.style 2 args = {%
      append after command=%
        {pic at (\tikzlastnode.\pgfkeysvalueof{/sacc/where}) {caccsq={#1}{#2}}}%
  },%
}

\def\markbaseline{-.33em}
\def\tacc#1{\tikz[baseline=\markbaseline]\pic{acc=#1};\xspace}
\def\tcacc#1#2{\tikz[baseline=\markbaseline]\pic{cacc={#1}{\upshape #2}};\xspace}

\def\emptyacc{\phantom{0}}

\tikzstyle{removed} = [opacity=0, overlay]
\tikzstyle{SCC} = [
  rounded corners,
  draw=black!50, very thin,
  fill=black!7
]
\tikzstyle{trivial} = [dashed]
\tikzstyle{sccname} = [red,anchor=south east,outer xsep=13pt]

\newcommand{\vh}[1]{\textcolor{orange}{\ifmmode \text{[#1]}\else [VH: #1] \fi}}
\newcommand{\ol}[1]{\textcolor{blue}{\ifmmode \text{[OL: #1]}\else [OL: #1] \fi}}
\newcommand{\bs}[1]{\textcolor{ForestGreen}{\ifmmode \text{[BS: #1]}\else [BS: #1] \fi}}
\newcommand{\nm}[1]{\textcolor{magenta}{\ifmmode \text{[NM: #1]}\else [NM: #1] \fi}}
\definecolor{oaviolet}{RGB}{138,43,226}
\newcommand{\oa}[1]{\textcolor{oaviolet}{\ifmmode \text{[OA: #1]}\else [OA: #1] \fi}}

\newcommand{\bigO}[0]{\mathcal{O}}
\newcommand{\bigOof}[1]{\bigO(#1)}

\newcommand{\aut}[0]{\mathcal{A}}
\newcommand{\autof}[1]{\aut_{#1}}
\newcommand{\states}[0]{\mathcal{Q}}
\newcommand{\colourset}[0]{\Gamma}
\newcommand{\trans}{\delta}
\newcommand{\transof}[1]{\trans_{#1}}
\newcommand{\inits}[0]{I}
\newcommand{\finals}[0]{F}
\newcommand{\colouring}[0]{\mathsf{p}}
\DeclareRobustCommand{\acccond}[0]{\mathsf{Acc}}

\newcommand{\buchi}[0]{B\"{u}chi\xspace}

\newcommand{\calM}[0]{\mathcal{M}}
\newcommand{\alphabet}[0]{\Sigma}
\newcommand{\word}[0]{w}
\newcommand{\wordof}[1]{\word_{#1}}

\newcommand{\ltr}[1]{\xrightarrow{#1}}

\newcommand{\lang}[0]{\mathcal{L}}
\newcommand{\langof}[1]{\lang(#1)}

\newcommand{\colours}[0]{\Gamma}
\newcommand{\colouringof}[1]{\colouring(#1)}

\newcommand{\macrostateand}[2]{%
	\begin{tikzpicture}[
			baseline=(macrostateandroot.base),
			level distance=4.5mm,
			sibling distance=7mm,
			macrostateleaf/.style={rectangle,draw,rounded corners=.6mm,inner sep=1.0pt,minimum width=0pt,minimum height=0pt,font=\scriptsize},
			macroroot/.style={rectangle,draw,rounded corners=.6mm,inner sep=1.0pt,minimum width=0pt,minimum height=0pt,font=\scriptsize}
		]
		\node[macroroot] (macrostateandroot) {$\wedge$}
			child { node[macrostateleaf] {$#1$} }
			child { node[macrostateleaf] {$#2$} };
	\end{tikzpicture}%
}

\newcommand{\GetSucc}[0]{\mathtt{Succ}}

\newcommand{\mstate}[0]{\mathcal{M}}
\newcommand{\GetSuccORF}[0]{\mathtt{SuccORF}}
\newcommand{\GetSuccSHB}[0]{\mathtt{SuccSHB}}
\newcommand{\IsORF}[0]{\mathtt{IsORF}}

\newcommand{\emersonlei}[0]{\mathbb{EL}}
\newcommand{\emersonleiof}[1]{\emersonlei(#1)}

\newcommand{\mytrue}[0]{\mathtt{true}}
\newcommand{\myfalse}[0]{\mathtt{false}}

\newcommand{\Inf}[0]{\mathsf{Inf}}
\newcommand{\Infof}[1]{\Inf(#1)}
\newcommand{\Fin}[0]{\mathsf{Fin}}
\newcommand{\Finof}[1]{\Fin(#1)}

\DeclareRobustCommand{\accinf}[0]{\textrm{Inf}}
\DeclareRobustCommand{\accfin}[0]{\textrm{Fin}}
\DeclareRobustCommand{\accinfof}[1]{\accinf(#1)}
\DeclareRobustCommand{\accfinof}[1]{\accfin(#1)}
\DeclareRobustCommand{\accinfcof}[1]{\accinf(\tacc{#1})}
\DeclareRobustCommand{\accfincof}[1]{\accfin(\tacc{#1})}
\DeclareRobustCommand{\accinfgcof}[2]{\accinf(\tcacc{#1}{$#2$})}
\DeclareRobustCommand{\accfingcof}[2]{\accfin(\tcacc{#1}{$#2$})}
\DeclareRobustCommand{\accinfcmof}[1]{\accinf(\scalebox{0.6}{\tacc{#1}})}
\DeclareRobustCommand{\accfincmof}[1]{\accfin(\scalebox{0.6}{\tacc{#1}})}
\DeclareRobustCommand{\accinfcmgof}[2]{\accinf(\scalebox{0.6}{\tcacc{#1}{$#2$}})}
\DeclareRobustCommand{\accfincmgof}[2]{\accfin(\scalebox{0.6}{\tcacc{#1}{$#2$}})}
\DeclareRobustCommand{\acccmgof}[2]{\scalebox{0.6}{\tcacc{#1}{$#2$}}}

\DeclareRobustCommand{\accinfidof}[1]{\accinf_{#1}(\tcacc{4}{$k$})}
\DeclareRobustCommand{\accinfidofsup}[1]{\accinf_{#1}(\scalebox{0.6}{\tcacc{4}{$k$}})}

\newcommand{\taccj}[0]{\tcacc{11}{$j$}}
\newcommand{\taccgof}[1]{\tcacc{7}{#1}}      
\newcommand{\taccBj}[0]{B_j}
\newcommand{\taccGj}[0]{G_j}
\newcommand{\taccGjl}[0]{G_{j,\ell}}

\newcommand{\IsSat}{\mathtt{IsSat}}

\newcommand{\SuccCtx}{\mathtt{SuccCtx}}
\newcommand{\restrict}[2]{#1{|}_{#2}}
\newcommand{\ctxW}{\mathit{W}}
\newcommand{\ctxR}{\mathit{R}}
\newcommand{\ctxP}{\mathit{P}}
\newcommand{\lst}{\mathbb{L}}

\newcommand{\kofola}{\textsc{Kofola}\xspace}
\newcommand{\spot}{\textsc{Spot}\xspace}
\newcommand{\spothigh}{\textsc{Spot-High}\xspace}
\newcommand{\kofolaind}{\textsc{Kofola-Ind}\xspace}
\newcommand{\kofolashbroot}{\textsc{Kofola-Ind-SHB-Root}\xspace}
\newcommand{\kofolashbsubtree}{\textsc{Kofola-Ind-SHB-Subtree}\xspace}
\newcommand{\kofolafor}{\textsc{Kofola-Ind-FOR}\xspace}

\newcommand{\colj}[0]{\tcacc{2}{j}}
\newcommand{\colk}[0]{\tcacc{4}{k}}

\newcommand{\complof}[1]{#1^{\sharp}}
\newcommand{\autcompl}[0]{\complof \aut}
\newcommand{\statescompl}[0]{\complof \states}
\newcommand{\transcompl}[0]{\complof \trans}
\newcommand{\initscompl}[0]{\complof \inits}
\newcommand{\acccondcompl}[0]{\complof \acccond}

\newcommand{\tree}{\textsf{Tr}}
\newcommand{\nil}{\textsf{NIL}}

\newcommand{\crset}[0]{\mathit{CR}}

\newcommand{\trOrAppendix}[1]{\ifTR{}\cref{#1}\else{}\cite{techrep}\fi}

\begin{document}

\maketitle

\begin{abstract}
 \emph{B\"{u}chi elevator automata} naturally appear in several areas of formal
  methods as a~structural expressibly-equivalent subclass of B\"{u}chi automata
  where every strongly connected component is either deterministic or
  inherently weak.
 It was shown that this class contains the majority of B\"{u}chi automata
 generated in practical applications, including LTL model-checking and
  verification of hyperproperties.
 Moreover, the elevator subclass enables more efficient complementation and determinization algorithms than unrestricted B\"{u}chi
 automata.
 In this paper, we introduce \emph{Emerson-Lei elevator automata}, which is
 a~generalization of B\"{u}chi elevator automata to richer acceptance
 conditions.
 We provide a~complementation algorithm with a~significantly better
 asymptotic complexity than the best known algorithm for unrestricted Emerson-Lei
 automata.
 The practical efficiency of our algorithm is demonstrated by an experimental comparison with the popular state-of-the-art tool Spot.
 Our work is,
 to the best of our knowledge, the first step towards practical algorithms for
 complementing, determinizing, and testing universality and inclusion of
 Emerson-Lei automata with rich acceptance conditions.
\end{abstract}

\vspace{-0.0mm}
\section{Introduction}\label{sec:introduction}
\vspace{-0.0mm}

The fundamental framework of automata on infinite words ($\omega$-automata)---such as \buchi automata (BAs), Rabin automata, or parity automata---has
proven successful for verifying software
properties~\cite{DBLP:conf/tacas/HeizmannCDGHLNM18,DBLP:conf/cav/HeizmannHP13,DBLP:journals/corr/abs-2102-01727,DBLP:conf/cav/HeizmannHP14,DBLP:conf/lics/VardiW86}.
Complementation of these automata is an operation used 
to analyze program termination~\cite{DBLP:conf/cav/HeizmannHP14},
model-check temporal logics (e.g., QPTL~\cite{DBLP:conf/lics/KestenP95}, HyperLTL~\cite{DBLP:conf/post/ClarksonFKMRS14}), 
and decide logics (e.g., S1S~\cite{buchi1962decision}, or fragments of first-order logic over Sturmian words~\cite{DBLP:journals/lmcs/HieronymiMOSSS24}).
Moreover, it is the underlying operation for determining whether the language
of one $\omega$-automaton is included in, or equivalent to, that of another.
Yet, in the general case, complementation is notoriously challenging, and its
difficulty has led the verification community to consider automata with
structural restrictions, which allow more efficient algorithms.
\emph{Elevator automata} are BAs whose maximal strongly connected components (SCCs) are all deterministic or inherently weak (i.e., all runs staying in the SCC are either accepting or rejecting).
This subclass was first identified in~\cite{HavlenaLS22a} as a generalization of
the well-known semi-deterministic (a.k.a.\ limit-deterministic)
automata~\cite{CourcoubetisY88}.
Recent work~\cite{AlexajHHLLM26} showed that the vast majority of BAs
considered in practical applications, e.g., LTL model checking and verification
of hyperproperties, are elevator automata.

Orthogonally, LTL model checking, synthesis, and other applications have raised interest in more sophisticated types of automata.
Many types of $\omega$-automata differ exclusively in their acceptance conditions~\cite{Boker18}.
The most general considered is that of \emph{Emerson-Lei automata} (ELAs)~\cite{EmersonL87}, whose acceptance condition is an arbitrary Boolean combination of $\Fin$ and $\Inf$ predicates.
The predicate $\Finof{\taccgof c}$ denotes that all transitions labeled with the color $\taccgof c$ occur only finitely often along an accepting run.
Dually, $\Infof{\taccgof c}$ denotes that some transition labeled with $\taccgof c$ occurs infinitely often.
For instance, BAs admit a single color $\tacc 0$, and their acceptance condition is $\Infof{\tacc 0}$.
The main motivation for using more complex acceptance conditions is the
succinctness of the automaton representation, since ELAs can be exponentially
more concise than BAs~\cite{SafraV89} (which is exploited already by some
LTL-to-automata translators like \texttt{ltl3tela}~\cite{ltl3tela}).
Another motivation is algorithmic flexibility and operational costs.
It is well known that BAs are not closed under determinization, whereas ELAs are; however, checking emptiness for ELAs is NP-complete, while it is in P for BAs~\cite{Boker18}.

Generalized Rabin automata (GRAs), whose acceptance condition is
$\bigvee_i \Finof{\tacc 0_i} \land \Infof{\tacc 1_i} \land \dots \land
\Infof{\colk_i}$\footnote{I.e., in order for a~run to be accepting, there needs to exist
a~clause $\Finof{\tacc 0_j} \land \Infof{\tacc 1_j} \land \dots \land
\Infof{\colk_j}$ in the disjunction such that the run sees only finitely many
occurrences of the color $\tacc 0_j$ and infinitely many occurrences of all
colors $\tacc 1_j, \ldots, \colk_j$.}, emerged from LTL model
checking~\cite{DBLP:conf/cav/KretinskyE12,DBLP:journals/fmsd/EsparzaKS16,DBLP:conf/cav/ChatterjeeGK13}
thanks to their flexibility and conciseness.
The practical use of GRAs highlights the need to refine the complexity of
complementing sophisticated automata types.
Recently, \cite{HavlenaLS25} improved the upper-bound when complementing ELAs and subclasses, including GRAs.
We aim for more efficient constructions that leverage the elevator structure, both by benefiting from its frequent occurrence and by concisely converting non‑elevator automata into elevator ones.
For a~given elevator automaton, \Cref{table:compare} compares the upper-bounds
on the size of its complement from best known algorithms for automata with
unrestricted structure and the upper-bounds for elevator automata provided in
this paper.
It is clear that exploiting the prevalence of elevator automata can have
a~substantial impact on reducing the size of the complement automata, in
particular, for ELAs, from double-exponential to single-exponential.

\begin{table}
  \caption{Upper bounds on the size of $\autcompl$, the complement of $\aut$
  having $n$ states and $k$ colors.}\label{table:compare}
  \vspace{-3mm}
  \centering
	\begin{tabular}{lll}
    \multirow{2}{*}{\textbf{Acceptance condition}} &
    \multirow{2}{40mm}{\centering \textbf{State-space of $\autcompl$ when $\aut$ is unrestricted}} &
    \multirow{2}{40mm}{\centering \textbf{State-space of $\autcompl$ when $\aut$ is elevator}} \\\\
    \midrule
		\buchi & $\bigOof{n (0.76n)^n}$~\cite{Schewe09} & $4^n$ \cite{HavlenaLLST23}
		\\
		gen. \buchi &  $\bigOof{n(0.76nk)^n}$~\cite{HavlenaLS25} & $(k + 3)^n$
		\\
		Rabin ($\ell$ pairs) & $\bigOof{n^{\ell}(0.76n)^{n\ell}}$~\cite{HavlenaLS25} & $3 \cdot 2^{(3\ell+2)n}$
		\\
		gen. Rabin ($\ell$ pairs) & $\bigOof{n^\ell (0.76kn)^{\ell n}}$~\cite{HavlenaLS25} & $3 \cdot 2^{(3k\ell+2)n}$
		\\
		Streett ($\ell$ pairs) & $2^{\Theta(n \log n + nk \log k)}$~\cite{CaiZ11a,CaiZ11b}  &  $6 \cdot \ell^{n} \cdot 2^{(2\ell+2)n}$
		\\
		Parity (min-odd index $\ell$) & $2^{\bigOof{n \log n}}$~\cite{CaiZ11b} & $3 \cdot 2^{(\frac{3}{2}\ell+2)n}$
		\\
		Emerson-Lei ($\ell$ atoms) &  $\bigOof{n^{2^k}(0.76nk)^{n2^k}}$~\cite{HavlenaLS25} & $3\ell \cdot 2^{(2\ell+3)n}$
	\end{tabular}
  \vspace*{-5mm}
\end{table}

The paper is organized as follows.
In \Cref{sec:preliminaries}, we introduce Emerson-Lei elevator and semi-deterministic automata.
In \cref{sec:rabin,sec:streett}, we describe the main ideas of our procedure on
two special cases of ELAs: semi-deterministic automata with
(i)~one generalized Rabin pair and
(ii)~one generalized Streett pair.
In \Cref{sec:inductive-compl}, we present an inductive procedure for semi-deterministic automata with an arbitrary Emerson-Lei condition.
In \Cref{sec:elevator-compl}, we extend our techniques from semi-deterministic
to elevator and unrestricted automata.
In particular, our conversion treats automata that are ``almost elevator''
(i.e., where only a~few SCCs are not elevator)
more efficiently.
This is helpful since in non-elevator automata from practice, most SCCs are elevator.
In \Cref{sec:experiments}, we evaluate our complementation procedure, implemented in \kofola~\cite{AlexajHHLLM26}, and compare its performance to \spot~\cite{Duret-LutzRCRAS22}.
Our benchmarks include ELAs from~\cite{ltl3tela}, randomly generated single-pair GRAs, and ELAs obtained from the translation of randomly generated LTL formulas.
We show that the complements constructed by \kofola are often much smaller than those produced by~\spot.

\vspace{0.0mm}
\section{Preliminaries}\label{sec:preliminaries}
\vspace{0.0mm}

We fix a~finite non-empty alphabet~$\Sigma$ and the first infinite
ordinal~$\omega$.
An (infinite) word~$\word$ is a~function $\word\colon \omega \to \Sigma$ where
the $i$-th symbol is denoted as~$\wordof i$.
Sometimes, we represent~$\word$ as an~infinite sequence $\word = \wordof 0
\wordof 1 \dots$
We denote the set of all infinite words over~$\Sigma$ as $\Sigma^\omega$;
an \emph{$\omega$-language} is a~subset of~$\Sigma^\omega$.
We use $\cdot$ for ellipsis, e.g., if interested only in the second component
\mbox{of a~triple, we may write the triple as $(\cdot, x, \cdot)$.}

\vspace{-2.0mm}
\subsection{Emerson-Lei Acceptance Conditions}
\vspace{-1.0mm}

Given a~set $\colourset = \{0, \ldots, k -1\}$ of~$k$
\emph{colors} (often depicted as \tacc{0}, \tacc{1}, etc.), we define the
set of \emph{Emerson-Lei acceptance conditions} $\emersonleiof \colourset$ as
the set of formulas constructed according to the following grammar where~$c$ ranges over $\colourset$:
\vspace{-2mm}
\begin{align*}
  \alpha ::= \mytrue \mid \myfalse \mid
             \Infof c \mid \Finof c \mid
             (\alpha \land \alpha) \mid (\alpha \lor \alpha).\\[-7mm]
\end{align*}
We denote by $|\alpha|$ the number of atomic conditions contained in $\alpha$,
where multiple occurrences of the same atomic condition are counted multiple
times.
Given an infinite sequence of sets of colors~$\gamma \colon \omega \to 2^{\colourset}$
and a~condition~$\alpha \in \emersonleiof \colourset$,
the \emph{satisfaction} relation $\models$ is defined inductively as follows (with $c \in \colourset$):
\vspace{-2mm}
\begin{align*}
  &\gamma \models \mytrue&& \gamma \models \alpha_1 \lor \alpha_2 \text{~ iff ~} \gamma \models \alpha_1 \text{ or } \gamma \models \alpha_2
	\\
  &\gamma \not\models \myfalse && \gamma \models \alpha_1 \land \alpha_2 \text{~ iff ~} \gamma \models \alpha_1 \text{ and } \gamma \models \alpha_2
	\\
  &\gamma \models \Finof c \text{~ iff ~} |\{ i \mid c \in \gamma(i)\}| < \infty &&
	\gamma \models \Infof c \text{~ iff ~} |\{ i \mid c \in \gamma(i)\}| = \infty.\\[-10mm]
\end{align*}


\newcommand{\tableConditions}[0]{
\begin{table}[t]
\caption{Popular acceptance conditions}
\label{tab:conditions}
\vspace*{-6mm}
\begin{center}
\begin{tabular}{lcc}
	\textbf{Name} & \textbf{Classic} & \textbf{Generalized}
	\\\midrule
	\emph{\buchi} & \small $\accinfof{\tacc 0}$ & \small $\bigwedge_{j=0}^{k-1} \accinfof{\taccj}$
	\\
	\emph{co-\buchi} \qquad & \small $\accfinof{\tacc 0}$ & \small $\bigvee_{j=0}^{k-1} \accfinof{\taccj}$
	\\
	Rabin pair & \small $\accfinof{B} \land \accinfof{G}$ & \small $\accfinof{B} \land \bigwedge_{\ell=0}^{m -1}\accinfof{G_\ell}$
	\\
	\emph{Rabin} & \small \qquad$\bigvee_{j=0}^{k-1} \accfinof{\taccBj} \land \accinfof{\taccGj}$ \qquad & \small \qquad$\bigvee_{j=0}^{k-1} (\accfinof{\taccBj} \land \bigwedge_{\ell=0}^{m_j -1}\accinfof{\taccGjl})$
	\\
	Streett pair & \small $\accinfof{G} \lor \accfinof{B} $ & \small $\accinfof{G} \lor \bigvee_{\ell=0}^{m -1}\accfinof{B_\ell}$
	\\
	\emph{Streett} & \small $\bigwedge_{j=0}^{k-1} \accinfof{\taccGj}\lor \accfinof{\taccBj}$ & \small $\bigwedge_{j=0}^{k-1} (\accinfof{G_j} \lor \bigvee_{\ell=0}^{m_j -1}\accfinof{B_{j,\ell}})$
	\\\midrule
	Parity & \multicolumn{2}{c}{\small $\accfinof{\tacc 0} \land (\accinfof{\tacc 1} \lor (\accfinof{\tacc 2} \land (\accinfof{\tacc 3} \lor (\accfinof{\tacc 4} \land \ldots))))$}
\end{tabular}
\end{center}
\vspace*{-8mm}
\end{table}
}

\vspace{0.0mm}
\subsection{Emerson-Lei Automata}\label{sec:ela}
\vspace{0.0mm}
%
A~(nondeterministic) transition-based\footnote{%
Extending our approach to state-based or mixed state and transition-based acceptance is straightforward.
}
\emph{Emerson-Lei automaton} (ELA)
over~$\Sigma$ is a~tuple $\aut = (\states, \trans, \inits, \colourset,
\colouring, \acccond)$,
where $\states$ is a~finite set of \emph{states},
$\trans \subseteq \states \times \Sigma \times \states$ is a~set of
\emph{transitions},
$\inits \subseteq \states$ is the set of \emph{initial} states,
$\colourset$ is the set of \emph{colors},
$\colouring\colon \trans \to 2^{\colourset}$ is a~\emph{coloring} of transitions, and
$\acccond \in \emersonleiof \colourset$ is the acceptance condition.
The negation of $\acccond$, i.e., $\neg\acccond$, is obtained by exchanging $\Inf(c)$ with $\Fin(c)$ and $\wedge$ with $\vee$.
We use $p \ltr a q$ to denote that $(p,a,q) \in \trans$ and sometimes
treat~$\trans$ as a~function $\trans\colon \states \times
\alphabet \to 2^{\states}$.
Moreover, we extend~$\trans$ to sets of states $P \subseteq \states$ as
$\trans(P, a) = \bigcup_{p \in P} \trans(p,a)$.
If $|\inits| \leq 1$ and $|\trans(q,a)| \leq 1$ for all
states $q \in \states$ and symbols $a \in \alphabet$, then $\aut$ is \emph{deterministic}.
We define $|\aut| = |\states|$ and use $n = |\states|$ in complexity bounds.

A~\emph{run}
of~$\aut$ from~$q \in \states$ on an input word~$\word$ is an infinite sequence $\rho\colon
\omega \to \states$ that starts in~$q$ and respects~$\trans$, i.e., $\rho(0) = q$ and
$\forall i \geq 0\colon \rho(i) \ltr{\wordof i}\rho(i+1) \in \trans$.
We define the infinite sequence of sets of colors $\colouring(\rho, w) \colon \omega\to 2^{\colourset}$ by $\colouring(\rho,w)(i) = \colouring(\rho(i), w_i, \rho(i+1))$.
A~run~$\rho$ on~$w$ is \emph{accepting} wrt an acceptance condition~$\alpha$ 
iff $\colouring(\rho, w) \models \alpha$. 
The \emph{language} of~$\aut$, denoted as~$\langof{\aut}$, is defined as the set of words
$w \in \alphabet^\omega$ for which there exists an accepting run $\rho$ over $w$ in~$\aut$ starting from some state in~$\inits$
such that $\colouring(\rho, w) \models \acccond$.
Popular acceptance conditions can be expressed in this more general framework
as given in \cref{tab:conditions}.
Furthermore, we use the syntactic sugar~$\aut = (\states, \trans, \inits, \finals)$ with~$F \subseteq \trans$ being a~set of \emph{accepting transitions} to
denote a~(transition-based) \buchi automaton (BA) that would be defined using the ELA definition
above as $(\states, \trans, \inits, \{\tacc 0\}, \{\tau \mapsto \emptyset \mid \tau
\in \trans \setminus \finals\} \cup \{\tau \mapsto \{\tacc 0\} \mid \tau \in \finals\}, \Infof{\tacc 0})$.

\tableConditions   

\vspace{-0.0mm}
\subsection{Emerson-Lei Elevator and Semi-deterministic Automata}
\vspace{-0.0mm}

Let $\aut = (\states, \trans, \inits, \colourset, \colouring, \acccond)$ be an~ELA.
A~non-empty set of states~$R \subseteq \states$ is a \emph{strongly connected
	component} (SCC)
of~$\aut$ iff there exists a~path of a non-zero length between every two states
of~$R$ and~$R$ does not admit a~superset with this property.\footnote{Note that all our SCCs are non-trivial and maximal.}
Let~$R$ be an SCC of~$\aut$.
We use $\transof R = \trans \cap (R \times \alphabet \times R)$
to denote the set of transitions internal to~$R$ and $\colouring_R = \{\tau
\mapsto \colouringof \tau \mid \tau \in \transof R\}$ for the restriction
of~$\colouring$ to~$\trans_R$.
We define $\autof R = (R, \transof R,\emptyset, \colourset, \colouring_R, \acccond)$
to be the ELA obtained by restricting~$\aut$ to
states from~$R$ and simplifying the acceptance condition in the standard way.\footnote{
That is, if a~color~$\tcacc 7 c$ does not have an occurrence in~$\autof R$, all
occurrences of $\Infof{\tcacc 7 c}$ in~$\acccond$ are substituted by $\myfalse$
and all occurrences of $\Finof{\tcacc 7 c}$ are substituted by $\mytrue$; the
resulting formula is then simplified using identity and annihilation rules to
remove all non-essential occurrences of $\mytrue$ and $\myfalse$.}
Observe that $\autof R$ does not have any initial state (this is not a problem
since we only talk about its structure).
The component~$R$ is \emph{inherently weak} iff for some state~$q \in R$, either
(i)~all runs in~$\autof R$ from~$q$ are accepting or
(ii)~no run in~$\autof R$ from~$q$ is accepting (we note that the particular choice of~$q$ is irrelevant).
$R$~is an \emph{elevator component} if it is one of the following kinds:
(i)~inherently weak,
(ii)~deterministic (i.e., $\autof R$ is deterministic), or
(iii)~generalized co-\buchi (i.e., the acceptance condition of~$\autof R$ is generalized co-\buchi).
An \emph{Emerson-Lei elevator automaton} (ELEA) is an ELA whose all SCCs
are elevator SCCs.

Let $\aut[q]$ for $q \in \states$ denote the ELA obtained from $\aut = (\states, \trans, {q}, \colourset, \colouring, \acccond)$ by removing states and transitions unreachable from $q$ (and simplifying $\colouring$ accordingly). We call $\aut$ an \emph{Emerson-Lei semi-deterministic automaton} (ELSA) if, for every state $q \in \states$, one of the following holds:
(i)~$q$ does not belong to any SCC,
(ii)~$q$ belongs to an SCC $R$ such that the acceptance condition of $\autof R$ is $\myfalse$, or
(iii)~$\aut[q]$ is deterministic, i.e., all runs starting from $q$ are deterministic.
Intuitively (and simplifying a~bit), runs in ELSAs start in nondeterministic
non-accepting components but after they have seen some color for the first
time, they cannot make nondeterministic choices any more---every accepting run
needs to stay in some deterministic accepting SCC.
If $\aut$ is an ELSA, we let $\states_n \subseteq \states$ denote the states
reachable only from SCCs with the acceptance condition $\myfalse$ (i.e.,
\emph{non-accepting} states), and $\states_d$ the remaining states (i.e.,
states in \emph{deterministic} accepting SCCs).

\vspace{-0.0mm}
\section{Complementing an ELSA with a Generalized Rabin Pair}\label{sec:rabin}
\vspace{-0.0mm}

We will introduce our complementation algorithm for ELEAs in several steps,
describing the main ideas on the simpler case of semi-deterministic automata
(\cref{sec:rabin,sec:streett,sec:inductive-compl})
and lifting the techniques to elevator automata in \cref{sec:elevator-compl}.
For semi-deterministic automata, we also start with describing the main
techniques using simpler acceptance conditions: generalized Rabin pair (this
section) and generalized Streett pair (\cref{sec:streett}).
Then, in \cref{sec:inductive-compl} we lift the ideas to ELEAs with an
arbitrary acceptance condition.

Let us start with the algorithm for complementing an ELSA $\aut = (\states,
\trans, \inits, \colourset, \colouring, \acccond)$ with a~single generalized
Rabin pair acceptance condition $\acccond = \Finof{\tacc{0}} \land
\Infof{\tacc{1}} \land \cdots \land \Infof{\colk}$. 
Our construction can be viewed as a combination of the Miyano-Hayashi
construction~\cite{MiyanoH84} and the NCSB algorithm for complementing
semi-deterministic BAs~\cite{BlahoudekHSST16}.
In the following, let~$\states_n$ and~$\states_d$ denote states in the
nondeterministic and deterministic parts, respectively.

For ELSAs, every accepting run eventually enters a~deterministic SCC and
remains there.  Hence, it suffices to decide acceptance once the run is
in the deterministic SCC.
Based on this observation, to complement $\Finof{\tacc{0}}$, it suffices to
check if runs in~$\states_d$ visit color~$\tacc{0}$ infinitely often.
Following the Miyano-Hayashi construction~\cite{MiyanoH84}, we use a pair
$(R,B)$ of sets of states: $R \subseteq \states_d$ (``reach'') tracks all current
runs in $\states_d$ (i.e., all possible states where~$\aut$ can be after
reading a given prefix of the input word), and $B \subseteq R$ (``breakpoint'')
contains those runs that have not yet seen~$\tacc{0}$.
Whenever a run in $B$ sees~$\tacc{0}$, it is removed from~$B$.
When~$B$ becomes empty, it is reset to~$R$, thereby restarting the inspection
of all runs, including newly entered ones.
Thus, $B$~is emptied infinitely often iff all runs in $R$ see
$\tacc{0}$ infinitely often.

To complement a~condition $\Infof{\colj}$, the NCSB
construction~\cite{BlahoudekHSST16} uses a~triple $(C, S, B)$ to detect whether
a~run in $\states_d$ sees~$\colj$ only finitely often, i.e., eventually stops
seeing $\colj$.
Intuitively, $(C, B)$ plays a role similar to $(R, B)$ in the Miyano-Hayashi
construction: the set~$C$ (``check'') tracks all current runs in $\states_d$,
while $B\subseteq C$ (``breakpoint'') contains those runs that will eventually
be moved out of $C$ to $S$.
In addition, the set $S \subseteq C$ (``safe'') collects runs that are guessed
to never see $\colj$ again.
At each step, runs in~$C$ are nondeterministically assigned either to remain
in~$C$ or move to~$S$.
Since $B\subseteq C$, the runs that move from~$C$ to~$S$ are also removed
from~$B$.
When~$B$ becomes empty, all inspected runs have been moved to~$S$ (or cannot
proceed) and we restart the inspection by setting $B = C$.
If a~run placed in~$S$ subsequently sees $\colj$, the previous guess was
incorrect and that branch will fail.
Since the runs in~$B$ are
deterministic, the number of tracked runs in~$B$ does not increase.
Moreover, there exists a correct guess for moving all runs in~$C$ to~$S$.
It is shown in~\cite{BlahoudekHSST16} that $B$ becomes empty infinitely often
iff all runs in $C$ see $\colj$ only finitely often.

We now present a construction for complementing an ELSA with a~single
generalized Rabin pair: $\acccond = \Finof{\tacc{0}} \land
\Infof{\tacc{1}} \land \cdots \land \Infof{\colk}$.
Our approach generalizes the Miyano-Hayashi construction and the NCSB
construction to this setting.
The goal is to ensure that every run over a word $w$ violates at least one of
these conditions.
We use a set $N$ to track runs in the nondeterministic part~$\states_n$.
For runs in~$\states_d$, instead of using a~pair $(R,B)$
for $\Fin$ and separate triples $(C_j,S_j,B_j)$ for each
$\Infof{\colj}$-condition, we share structures: a single set~$C$ tracks runs
for all $\Inf$-conditions, and a~single breakpoint set~$B$ is used globally.

Accordingly, a macrostate in the complement automaton is a tuple $(N, R, C, S_1, \ldots, S_k, B)$.
Intuitively, runs in $\states_n $ are stored in $N$.
Upon entering $\states_d$, a run is nondeterministically placed in $R$ (to
violate $\Finof{\tacc{0}}$) or in $C$ (to violate some $\Inf$-condition).
All runs under inspection within $R \cup C$ are initially placed in $B$.
For a run in $C \cap B$ that is guessed to violate $\Infof{\colj}$, it is moved
to $S_j$ and removed from $B$. For runs in $R \cap B$, we remove runs visiting
$\tacc{0}$ from~$B$. 
When~$B$ becomes empty, all runs in $R \cup C$ have been inspected and~$B$ is
reset to $R \cup C$ again.
Thus, if $B$ is emptied infinitely often and no branch is rejected due to a
wrong guess, every run over $w$ violates at least one condition of the
generalized Rabin pair.

Formally, we define the complement BA $\autcompl =
(\statescompl, \transcompl, \initscompl, \acccondcompl)$ where the following
holds:
\begin{itemize}
  \item The set $\statescompl$ of macrostates consists of tuples $(N, R, C,
    S_1, \ldots, S_k, B) \in 2^{\states_n} \times (2^{\states_d})^{k+3}$ with
    the restrictions that
    (i)~$R \cap (C \cup S_1 \cup \cdots \cup S_k) = \emptyset$ (no run is
        tracked by more than one procedure),
    (ii)~$C \cap S_i = \emptyset$ for all $1\leq i \leq k$ (safe runs do not need to be checked any more),
    (iii)~$S_i \cap S_j = \emptyset$ for all $1 \leq i < j \leq k$ (safe runs
        need to break only one condition), and
    (iv)~$B \subseteq R \cup C$ (breakpoint tracks only relevant runs).

  \item The set $\initscompl \subseteq \statescompl$ of initial macrostates is
    defined by $\initscompl = \{ (\inits \cap \states_n, R, C, \emptyset, \ldots,
    \emptyset, B) \mid R \cap C = \emptyset, B = R \cup C = \inits \cap
    \states_d\}$.
  That is, initially, all runs starting in $\states_d$ are nondeterministically
    partitioned into $R$ and $C$, and all of them are placed in the set~$B$.

  \item For a macrostate $(N, R, C, S_1, \ldots, S_k, B)\in\statescompl$ and an
    input symbol~$a$, we have that $(N', R', C', S'_1, \ldots, S'_k, B') \in
    \transcompl((N, R, C, S_1, \ldots, S_k, B), a)$ if the following holds:

	\begin{itemize}
		\item $N' = \trans(N, a) \cap \states_n$ (The set $N$ tracks runs in the nondeterministic component.),

    \item $\trans(N, a) \cap \states_d \subseteq R' \cup C' \cup S'_1 \cup \cdots \cup S'_k$
      (Every run entering the deterministic component from $N$ must be
      tracked by either the $\Fin$-violating or the $\Inf$-violating
      procedure.),

    \item $\trans(S_j, a) \subseteq S'_j$ and $S'_j \setminus \trans(S_j, a)
      \subseteq \trans(C, a)$ for each $\tacc 1 \leq \colj \leq \colk$ (Each set $S_j$
      stores the runs that are guessed to avoid color $\colj$ from now on.
      Accordingly, $\colj$-safe runs stay $\colj$-safe and, additionally, runs
      from~$C$ can become $\colj$-safe.)

	\item $\colj \notin \colouring(q, a, \trans(q, a))$ for each $\tacc 1 \leq \colj \leq \colk$ and $q \in S_j$
    (Every run already contained in $S_j$ must avoid color $\colj$ on the current transition.),
		
		\item $B' = R' \cup C'$ if $B = \emptyset$, otherwise $B' \subseteq \trans(B, a)$
      (Restart inspection for all runs in $R\cup C$ when $B$ is empty.);
      moreover, each run leaving~$B$ must be correctly discharged: 

      \begin{itemize}
        \item if a run from $B \cap R$ is removed from $B'$, then it must have seen color $\tacc{0}$, i.e., 
          for all $q \in B \cap R$ such that $\trans(q,a) \notin B'$, we have $\tacc 0 \in
          \colouring(q, a, \trans(q, a))$;
        \item if a run from $B \cap C$ is removed from $B'$, then its successor
          must be placed into some safe set $S'_j$, i.e., for all $q' \in
          \trans(B \cap C, a) \setminus B'$, we have $q'
          \in S'_j$ for some $\tacc 1 \leq \colj \leq \colk$.
      \end{itemize}
		
    \item The set $R$ tracks runs that are guessed to violate
      $\Finof{\tacc{0}}$, i.e., runs that see the color $\tacc{0}$ infinitely often.
      Hence, $R$ is closed under deterministic successors and may additionally
      receive runs entering $\states_d$ from $N$:
	$\trans(R, a) \subseteq R' , R' \setminus \trans(R, a) \subseteq \trans(N, a) \cap \states_d$. 

    \item The set $C$ tracks runs that will violate one of the
      $\Inf$-conditions. Such runs evolve deterministically and are moved to
      some safe set, and $C$ may additionally receive runs entering $\states_d$
      from $N$: $\trans(C, a) \setminus (S'_1 \cup \ldots \cup S'_k) \subseteq
      C'$ and $C' \setminus \trans(C, a) \subseteq \trans(N, a) \cap
      \states_d$.
	\end{itemize}
	
	\item The \buchi acceptance condition requires the breakpoint set to be reset infinitely often.
    Thus, $\acccondcompl = \{(\cdot, \ldots, \cdot, B) \ltr a (\cdot, \ldots,
    \cdot, \cdot) \in \transcompl \mid B = \emptyset\}$.\footnote{
    We remind the reader of the abuse of notation where the acceptance
    condition for BAs is a~set of accepting transitions (\cref{sec:ela}).}
\end{itemize}



\noindent
Note that the construction does not allow merging of runs tracked by different
procedures, e.g., in the case the successors of a~state $r \in R$ and a~state
in~$c \in C$ are the same state~$q$, the construction cannot continue.
This is not a~problem since, due to nondeterminism, there will also be a~run
where the predecessor of~$c$ that was nondeterministically put to~$C$ is put
to~$R$ instead, and also a~run where the predecessor of~$r$ that was
nondeterministically put to~$R$ is put to~$C$ instead (and similarly for
merging runs between different $S$-sets).

\begin{theorem}\label{thm:complement-gr-sdela}
$\lang(\autcompl) = \Sigma^\omega\setminus \lang(\aut)$
\end{theorem}

\begin{proof}[Proof (sketch)]
Every word $w \in \lang(\autcompl)$ is not accepted by $\aut$. Indeed, along any accepting macrorun of $\autcompl$, the set $B$ becomes empty infinitely often, which implies that all runs in $R \cup C$ violate at least one condition of the generalized Rabin pair.

Conversely, let $w \notin \lang(\aut)$. Then every run of $\aut$ over $w$ violates at least one condition of the generalized Rabin pair. Hence, there exists a correct nondeterministic guess that assigns each run to the set corresponding to the condition it violates, at the earliest point. Under this guess, the macrorun of $\autcompl$ over $w$ does not terminate. Moreover, every run in $B$ is eventually removed—either by visiting $\tacc{0}$ or by being moved to some $S_i$—so that $B$ becomes empty infinitely often. Therefore, the macrorun is accepting in $\autcompl$, i.e., $w \in \lang(\autcompl)$.
\end{proof}

Let us now consider the size of~$\autcompl$.
We can deduce the bound on the number of states of~$\autcompl$ by checking, for
each state $q \in \states$, in which internal sets of a~macrostate $\calM =
(N,R,C, S_1, \ldots, S_k, B)$ it can be and give each such a~choice
a~number~$f(q)$ as follows:
\begin{itemize}
  \item  If~$q$ is not present in~$\calM$ at all, we assign $f(q) = 1$.
  \item  If $q \in \states_n \cap N$ or $q \in \states_d \cap R \cap B$, we assign
    $f(q) = 2$ (i.e., we can merge~$N$ and $R\cap B$ into one set~$W$ and split it
    into~$N$ and~$R \cap B$ anytime as $N = W \cap \states_n$ and $R\cap B = W
    \cap \states_d$).
  \item  If $q \in R \setminus B$, we let $f(q) = 3$,
    if $q \in C \cap B$, we let $f(q) = 4$,
    if $q \in C \setminus B$, we let $f(q) = 5$.
  \item  If $q \in S_j$ for $\tacc 1 \leq \colj \leq \colk$ , we assign $f(q) = 5+j$.
\end{itemize}
The total number of functions $f\colon \states \to \{1, \ldots, 5+k\}$ is
$(5+k)^n$ where $n = |\states|$.

\begin{theorem}\label{thm:rabin-size}
The size of $\autcompl$ is bounded by $(5+k)^n$.
\end{theorem}

We contrast the obtained state complexity $(5+k)^n$ to the best upper bound we
are aware of for unrestricted generalized Rabin pair ELAs, which is
$\bigOof{nk^n(0.76n)^n}$~\cite{HavlenaLS25}.
Exploiting the automaton structure can indeed have a profound effect on the
hardness of complementation.

Note that if the ELSA has the generalized \buchi condition (i.e., it does not
contain $\Fin$), we can remove the~$R$ part of the macrostate.
Then we can omit considering $R\cap B$ and $R \setminus B$ in the
complexity analysis and we obtain the size of the output bounded by $(3+k)^n$.

\vspace{-0.0mm}
\subsection{Optimization using Safe Run Fall-Through}\label{sec:rabin-safe-run}
\vspace{-0.0mm}

One practical issue of the procedure given above is a~high degree of
nondeterminism, since at every step,
(i)~newly incoming runs to~$\states_d$ need to be moved
    nondeterministically to either the $R$-part or the $(C, S_1, \ldots,
    S_k)$-part of the macrostate and, at the same time,
(ii)~runs in~$B$ need to either stay in~$B$ or nondeterministically move to
    some~$S_i$ (for every such run, all~$S_i$'s need to be considered)
causing a high out-degree of the constructed macrostates.
We give a~modification of the procedure that deals with point~(ii) as follows:
\begin{enumerate}
  \item  We either move all runs from~$B$ to~$S_1$ or let them all stay in~$B$.
  \item  If a~run in any~$S_j$ sees~$\colj$, we move it to~$S_{j+1}$ (unless
    $j=k$ in which case we do not generate the successor).
  \item  If a~run in~$S_i$ is merging with a~run in~$S_j$ for $j > i$, we keep
    the run in~$S_j$.
\end{enumerate}
Intuitively, we nondeterministically decide a~point from which some~$\Inf$
condition is violated by a~run and then try, one by one, to find which $\Inf$
condition it is.
We give more details about the construction in the inductive procedure in
\cref{sec:or-fin-opt}.



\vspace{-0.0mm}
\section{Complementing an ELSA with a Generalized Streett Pair}\label{sec:streett}
\vspace{-0.0mm}

Next, we give our algorithm for complementing an ELSA $\aut = (\states,
\trans, \inits, \colourset, \colouring, \acccond)$ with a~single generalized
Streett pair condition $\acccond = \Infof{\tacc{0}} \lor \Finof{\tacc{1}} \lor
\cdots \lor \Finof{\colk}$, which can be seen as a~dual of the generalized
Rabin pair condition described in the previous section.

For a~generalized Streett pair, it suffices to ensure that every run violates
\emph{all} conditions in $\acccond$.
As before, our algorithm builds on top of the Miyano-Hayashi~\cite{MiyanoH84}
(to handle the $\Fin$ conditions) and NCSB~\cite{BlahoudekHSST16} (to handle
the $\Inf$ conditions).
A~macrostate in the complement automaton has the structure $(N,\crset,S,B,z)$,
where~$N \subseteq \states_n$ tracks runs in the nondeterministic part and $\crset,
S, B \subseteq \states_d$ track runs in the deterministic parts; $z$~is used as
a~scheduling counter for checking that all $k+1$ acceptance conditions are
invalidated in a~round robin manner.
In essence, the algorithm runs one instance of NCSB and $k$~instances of the
Miyano-Hayashi algorithm in sequence, reusing the sets in the macrostate
between them.
In the macrostate, $\crset$~is used to hold the content of the $C$-set of NCSB
(when $z=0$) or the $R_j$-set of the Miyano-Hayashi algorithm for invalidating
$\Finof{\colj}$ (when $z=j > 0$).

Initially, $B=\crset$ and $z=0$, indicating that the condition
$\Infof{\tacc{0}}$ is currently under inspection.
The breakpoint~$B$ contains exactly those runs that are being checked for
violating the currently scheduled condition.
When $z=0$, for violating the $\Infof{\tacc{0}}$ condition, we need to
nondeterministically move all runs from~$B$ to~$S$ eventually (runs in~$S$ are
not allowed to see~$\tacc 0$ any more) before switching to the next condition.
For $z=j \geq 1$, we are inspecting the condition $\Finof{\colj}$.
In this phase, a~run remains in~$B$ until it sees color $\colj$, at which point
it is removed from~$B$, witnessing a violation of the $\Fin$-condition.
When~$B$ becomes empty, we increment~$z$ and switch to the next condition.
At every switch, we re-sample $B=\crset$.

This cyclic inspection ensures that every run is repeatedly checked against
each condition.
If~$B$ becomes empty infinitely often, then for every run and for every condition in $\acccond$, there is a point at which the run is verified to violate that condition.
Hence, all runs violate all conditions, and thus the generalized Streett pair condition.

Formally, we define the complement BA $\autcompl = (\statescompl, \transcompl, \initscompl, \acccondcompl)$ as follows.
\begin{itemize}
  \item The set $\statescompl$ of macrostates contains tuples $(N, \crset, S,
    B, z) \in 2^{\states_n}  \times (2^{\states_d})^{3}\times \{0, \ldots, k\}$
    with $S \subseteq \crset$ and $B \subseteq \crset$.
  \item $\initscompl = \{ (\inits \cap \states_n, \inits \cap \states_d,
    \emptyset, \inits \cap \states_d, 0)\}$ contains a~single initial macrostate
    (we start with checking the $\Inf$ condition).
	
	\item For a macrostate $(N, \crset, S, B, z)$ and $a\in \alphabet$, $(N', \crset', S', B', z') \in \transcompl((N, \crset, S, B, z),a )$ if the following holds:
	\begin{itemize}
		\item $N' = \trans(N, a) \cap \states_n$ ($N$ tracks the runs in the non-deterministic component.),
    \item $\crset' = \trans(N \cup \crset, a) \cap \states_d$ ($\crset$ tracks
      all runs in~$\states_d$ and those that came from~$\states_n$.),
    \item  $\tacc 0 \notin \colouring(q, a, \trans(q, a))$ for each $q \in S$
      (Safe states cannot see~$\tacc 0$ any more.),
    \item  if $B\neq \emptyset$ (We keep tracking the current condition.) then $z' = z$ and
      \begin{itemize}
        \item  if $z = 0$ (we are tracking $\Inf(\tacc 0)$):
          \begin{itemize}
            \item  $S' = \trans(S, a)$ and $B' = \trans(B,a) \setminus S'$ (We wait with moving~$B$ to~$S$.) or
            \item  $S' = \trans(S \cup B, a)$ and $B' = \emptyset$ (We move runs in~$B$ to~$S$.),
          \end{itemize}
        \item  if $z = j > 0$ (we are tracking $\Fin(\colj)$), then
          $S' = \trans(S, a)$ and $B' = \trans(\{q \in B \mid \colj \notin
          \colouring(q,a, \trans(q,a))\}, a)$ (We remove from~$B$ runs that see~$\colj$.),
      \end{itemize}
    \item  if $B=\emptyset$ then $S' = \trans(S, a)$, $z' = (z+1)\bmod (k+1)$, and
      $B' = \crset' \setminus S'$ (if $z'=0$) or
      $B' = S'$ (if $z> 0$).

	\end{itemize}
	
  \item $\acccondcompl = \{(\cdot, \ldots, \cdot, B) \ltr a (\cdot, \ldots,
    \cdot, \cdot) \in \transcompl \mid B = \emptyset\}$ (We need to reset~$B$
    infinitely often).
\end{itemize}

Note that the construction has (unlike many other complementation constructions
for omega automata) a quite high degree of determinism: each
macrostate has for every symbol at most two successors (when $z=0$) or at most
one successor (when $z > 0$).
We also note that we are sampling the breakpoint~$B$ with~$\crset$ only when
entering the $\Inf(\tacc 0)$ part of the algorithm, when moving to~$\Fin$
conditions, we sample~$B$ with~$S$.
We could also sample $B$ with $\crset$, but this would make the complexity of
the procedure worse.

\begin{theorem}\label{thm:complement-gs-sdela}
$\lang(\autcompl) = \Sigma^\omega\setminus \lang(\aut)$
\end{theorem}

\begin{proof}[Proof (sketch)]
In a macrorun over a word $w \in \langof{\autcompl}$, the set $B$ becomes empty infinitely often.
This means that the counter $z$ cycles through all conditions. Hence, every run is repeatedly inspected and shown to violate each condition in $\acccond$. Thus, all runs in $\aut$ violate all conditions, and so $w \notin \lang(\aut)$.

Conversely, let $w \notin \lang(\aut)$. Then every run over $w$ violates all conditions in $\acccond$. There exists a correct nondeterministic guess such that, in each phase, every run in $B$ is eventually removed—either by being moved to $S$ or by witnessing the corresponding $\Fin$-condition—so that $B$ becomes empty infinitely often. Since all conditions are inspected cyclically, each violation is eventually verified. Hence, the macrorun is accepting, and $w \in \lang(\autcompl)$.
\end{proof}

The following theorem gives bounds on the size of~$\autcompl$ (obtained in
a~similar manner as the bounds in \cref{thm:rabin-size} and proved in
\trOrAppendix{sec:proof-thm-streett-size}).

\begin{theorem}\label{thm:streett-size}
The size of $\autcompl$ is bounded by $(k+1)\cdot 4^n$.
\end{theorem}

Note that if there is no $\Inf$ condition, i.e., the acceptance condition is
generalized co-\buchi, then we do not need the~$S$ part of the macrostate;
$\autcompl$ then has at most $k\cdot 3^n$ states.

\newcommand{\algSampling}[0]{
\begin{figure}[t]
\begin{algorithm}[H]
  \caption{Sampling successor runs $\GetSucc(M_{\psi}, a, C, f)$}
  \label{alg:induc-succ-tree}

  \If{$\psi = \accfingcof{5}{\ell} $ with $ e= S_{\accfincmgof{5}{\ell}}$}{
    $L = 
		\emptyset $ if $\exists t = s \ltr a s' \in \delta$ for $ s \in S \cup C$ s.t $\tcacc{5}{$\ell$} \in \colouring(t)$ and otherwise $L =
		\{\delta(S \cup C, a) \}$\;
	$M = \{\tree(\accfingcof{5}{\ell}, L_{\accfincmgof{5}{\ell}}, \nil, \nil)\}$\;
  } \uElseIf{$\psi = \accinfgcof{4}{k} $ with $ e= (H, B)_{\accinfcmgof{4}{k}}$} {
    $L = (\trans(H\cup C, a), \trans(H\cup C, a))$ if $f = \top$ otherwise $L = (\trans(H, a), \delta(B,a) \setminus \{t \in \delta \mid \tcacc{4}{$k$} \in \colouring(t)\})$\;
    $M = \{ \tree(\accinfgcof{4}{k}, L_{\accinfcmgof{4}{k}}, \nil, \nil)\}$\;
  }\uElseIf{$\psi = \varphi_1 \wedge \varphi_2 $} {
  	$M = \{ \tree(\wedge, -, m_1, m_2) \mid m_1 \in \GetSucc(M_{\varphi_1}, a, C, f), m_2 \in \GetSucc(M_{\varphi_2}, a, C, f)\}$\;
  } \uElseIf {$\psi = \varphi_1 \vee \varphi_2$}{
   $M = \{ \tree(\vee, - , m_1, m_2) \mid m_1 \in \GetSucc(M_{\varphi_1}, a, C_1, f), m_2 \in \GetSucc(M_{\varphi_2}, a, C_2, f), C = C_1\cup C_2\}$ \tcp*[r]{Nondet.\ split of $C$}
  }
\Return $M$
\end{algorithm}
\vspace*{-5mm}
\end{figure}
}

\vspace{-0.0mm}
\section{Inductive Procedure for Complementing Deterministic Components}\label{sec:inductive-compl}
\vspace{-0.0mm}

This section generalizes the constructions of \Cref{sec:rabin,sec:streett} to arbitrary acceptance conditions. Let $\aut = (\states_n \cup \states_d, \trans, \inits, \colours, \colouring, \acccond)$ be a semi-deterministic ELA. To complement $\aut$, we require that every run entering $\states_d$ satisfies $\psi = \neg\acccond$. 
As before, we use a set $N$ to track the runs in $\states_n$ and a set $C$ for runs entering $\states_d$.
Unlike generalized Rabin and Streett pairs, $\psi$ may contain both conjunctions and disjunctions, which complicates the run tracking.
To address this, we distinguish two modes:
\emph{sampling} and \emph{checking}.
In the sampling mode, we nondeterministically guess, for each run in $\states_d$, which disjuncts of each $\vee$-subformula in $\psi$ it will satisfy.
Accordingly, each subformula of $\psi$ is associated with a set of runs expected to satisfy it.
In the checking mode, we verify these guesses.
For a~condition $\accinfgcof{4}{k}$, we use a pair $(H,B)$, where $H$ contains the runs guessed to satisfy the condition and $B \subseteq H$ is a~breakpoint set used to verify that these runs visit $\tcacc{4}{$k$}$ infinitely often. Concretely, a run is removed from $B$ once it visits $\tcacc{4}{$k$}$, and whenever $B$ becomes empty, it is reset to $H$.
 For a~condition $\accfingcof{5}{\ell}$, it suffices to ensure that runs in $S$
 do not visit $\tcacc{5}{$\ell$}$.

To implement this, we employ \emph{tree-macrostates} that mirror the syntactic structure of $\psi$. Formally, a tree-macrostate $M_{\psi}$ is a tree constructed by $\tree(r, e, m_1, m_2)$, where the root label $r$ is either $\wedge$, $\vee$, or a basic condition of the form $\accinfgcof{4}{k}$ or $\accfingcof{5}{\ell}$, $e$ is a node label, and $m_1, m_2$ are the tree-macrostates corresponding to the subformulas (or the empty tree $\nil$ if absent).
For internal nodes, $e$ is undefined (denoted by $-$). For leaf nodes, $e$ stores sets of runs: $e = (H,B)_{\accinfcmgof{4}{k}} \subseteq \states_d^2$ for $\Inf$-conditions and $e = S_{\accfincmgof{5}{\ell}} \subseteq \states_d$ for $\Fin$-conditions.
We~write a leaf node simply as its label~$e$.

Intuitively, in the sampling mode, we ensure that each run satisfies $\psi$ by propagating it through the tree. At a leaf, a run is assigned to $H$ (for $\Inf$-conditions) or to $S$ (for $\Fin$-conditions). If $\psi = \varphi_1 \wedge \varphi_2$, the run must satisfy both subformulas and is therefore assigned to both subtrees $M_{\varphi_1}$ and $M_{\varphi_2}$. If $\psi = \varphi_1 \vee \varphi_2$, the run must satisfy at least one subformula, and is thus nondeterministically assigned to either~$M_{\varphi_1}$ or~$M_{\varphi_2}$.

In the checking mode, the construction proceeds as in the previous sections. The key idea is that breakpoint~$B$ is emptied infinitely \mbox{often iff all runs are correctly guessed and satisfy~$\psi$.}


Formally, let $\aut = (\states_n \cup \states_d, \trans, \inits, \colours, \colouring, \acccond)$ be a semi-deterministic ELA whose negated acceptance condition is $\psi$.
We define an NBA $\autcompl = (\statescompl, \transcompl, \initscompl, \acccondcompl)$ as follows:
\begin{itemize}
\item The set of macrostates $\statescompl$ contains macrostates of the form $(N,C,\mstate_\psi)$ where $N \subseteq \states_n$, $C \subseteq \states_d$, and $\mstate_\psi$ is a tree-macrostate.
\item The initial set of macrostates is $\initscompl = \{ (\inits \cap \states_n, \inits \cap \states_d, \mstate^0_\psi) \}$, where $\mstate^0_\psi$ has the same syntactic structure as $\psi$ and in particular, $M^0_{\accfincmgof{5}{\ell}} = \tree(\accfingcof{5}{\ell}, \emptyset_{\accfincmgof{5}{\ell}}, \nil, \nil)$  for each $\accfingcof{5}{\ell}$ and $M^0_{\accinfcmgof{4}{k}} = \tree(\accinfgcof{4}{k}, (\emptyset, \emptyset)_{\accinfcmgof{4}{k}}, \nil, \nil)$ for each $\accinfgcof{4}{k}$. Intuitively, all deterministic runs are initially placed in the $C$-set, and the automaton has not yet entered the sampling-and-checking mode.
\item The definition of $\transcompl$ and $\acccondcompl$ will be defined in detail below.

\algSampling   

\end{itemize}


\begin{figure}[b]\centering
  \vspace*{-4mm}
	\begin{subfigure}[b]{0.19\linewidth}\centering
\begin{tikzpicture}[automaton,scale=0.6,transform shape]
  \node[state, initial, initial above] (p) {$p$};
  \node[state, shift={(.8cm,-1.2cm)} ] (r) at (p) {$r$};
  \node[state, shift={(-.8cm,-1.2cm)}] (q) at (p) {$q$};

  \path[->]
    (p) edge             node[above right] {$a$} pic[auto]{acc=1} (r)
    (p) edge             node[above left] {$a$} (q)
    (r) edge[loop below] node[below] {$a$} pic[auto]{acc=3} (r)
    (q) edge[loop below] node[below] {$a$} pic[auto]{acc=2} pic[auto,pos=0.1]{acc=1} (q);
\end{tikzpicture}
		\caption{Automaton}
	\end{subfigure}
	\hfill
	\begin{subfigure}[b]{0.32\linewidth}\centering
		\resizebox{\linewidth}{!}{
\begin{tikzpicture}[
  level distance=9mm,
  level 1/.style={sibling distance=28mm},
  level 2/.style={sibling distance=14mm},
  tinner/.style={draw, rectangle, rounded corners=1.5mm, inner sep=2pt, font=\small},
  tleaf/.style={draw, rectangle, rounded corners=1.5mm, inner sep=2pt, font=\scriptsize},
  edge from parent/.style={draw, -}
]
  \node[tinner] {$\wedge$}
    child { node[tleaf] {$(\{r\},\emptyset)_{\accinfcmof{3}}$} }
    child { node[tinner] {$\vee$}
      child[xshift=-5mm] { node[tleaf] {$\{r\}_{\accfincmof{1}}$} }
      child { node[tleaf] {$(\{q\},\emptyset)_{\accinfcmof{2}}$} }
    };
\end{tikzpicture}}
		\caption{Some tree-macrostate $\mathcal{M}_{\psi}$}
	\end{subfigure}
	\hfill
	\begin{subfigure}[b]{0.45\linewidth}\centering
		\resizebox{\linewidth}{!}{\begin{tikzpicture}[
  level distance=8mm,
  level 1/.style={sibling distance=22mm},
  level 2/.style={sibling distance=13mm},
  tinner/.style={draw, rectangle, rounded corners=1.5mm, inner sep=2pt, font=\small},
  tleaf/.style={draw, rectangle, rounded corners=1.5mm, inner sep=2pt, font=\scriptsize},
  tleafhi/.style={draw, rectangle, rounded corners=1.5mm, inner sep=2pt,
                  font=\scriptsize, fill=red!25},
  edge from parent/.style={draw,-}
]
  \node[tinner] (t1) at (0,0) {$\wedge$}
    child { node[tleaf] {$(\{r,q\},\{r,q\})_{\accinfcmof{3}}$} }
    child { node[tinner] {$\vee$}
      child[xshift=-5mm] { node[tleaf] {$\{r\}_{\accfincmof{1}}$} }
      child { node[tleaf] {$(\{q\},\{q\})_{\accinfcmof{2}}$} }
    };

  \node[tinner] (t2) at (4.5cm,0) {$\wedge$}
    child { node[tleaf] {$(\{r,q\},\{r,q\})_{\accinfcmof{3}}$} }
    child { node[tinner] {$\vee$}
      child[xshift=-8mm] { node[tleafhi] {$\{r,q\}_{\accfincmof{1}}$} }
      child { node[tleaf] {$(\{q\},\{q\})_{\accinfcmof{2}}$} }
    };
\end{tikzpicture}}
		\caption{Tree-macrostates of $\GetSucc(\mstate_{\psi}, a, \{q\}, \top)$}
	\end{subfigure}
  \vspace{-2mm}
	\caption{$\GetSucc$ computation with the negated condition $\psi = \accinfcof{3} \wedge (\accfincof{1} \vee \accinfcof{2})$.}
	\label{fig:getsucc-example}
\end{figure}


To construct $\transcompl$ and $\acccondcompl$, we first introduce the notions of successors and satisfiability of a given tree-macrostate.
We define the successor construction for tree-macrostates (see~\cref{alg:induc-succ-tree}).
Let $\mstate_\varphi$ be a tree-macrostate for a subtree $\varphi$ of $\psi$, $a \in \Sigma$, $C \subseteq \states_d$, and $f \in \{\bot,\top\}$ a mode flag, where $\top$ denotes sampling mode.
The set of successor tree-macrostates is defined inductively as specified in \Cref{alg:induc-succ-tree}.
Observe that, for $\Fin$-leaves, the presence of a $\tcacc{5}{$\ell$}$-labelled transition eliminates all successors.
Moreover, the flag $f$ is only relevant for $\Inf$-leaves: if $f = \top$, the breakpoint set $B$ is resampled; if $f = \bot$, it continues to track runs that have not visited $\tcacc{4}{$k$}$.
%

%

\begin{example}\label{ex:getsucc}

  Consider the automaton in \Cref{fig:getsucc-example}(a) with negated acceptance condition $
\psi = \accinfcof{3} \wedge (\accfincof{1} \vee \accinfcof{2})$.
Consider the tree-macrostate $\mstate_\psi$ in \Cref{fig:getsucc-example}(b), with $C = \{q\}$ and $f = \top$.
At the $\wedge$-node, both $C$ and $f$ are forwarded to the two children.
For the $\accinfcof{3}$-leaf, we have $H=\{r\}$ and $B=\emptyset$. Its successors are trees $\tree(\accinfcof{3},(\{r,q\}, \{r,q\})_{\accinfcof{3}}, \nil, \nil)$,
since $\delta(H\cup C, a) = \delta(\{r,q\}, a) = \{r,q\}$ and $f = \top$ (hence $B$ is resampled).
At the $\vee$-node, $C$ is nondeterministically split into $C_1 \cup C_2$.
If $C_1=\emptyset$ and $C_2=\{q\}$, then at the $\accfincof{1}$-leaf, we have $\delta(\{r\},a)=\{r\}$, as no state in $S \cup C_1 = \{r\}$ admits a $\tcacc{1}{1}$-colored transition. 
The $\accinfcof{2}$-leaf resamples its breakpoint (since $f=\top$) and returns trees $\tree(\accinfcof{2}, (\{q\}, \{q\})_{\accinfcof{2}}, \nil, \nil)$,
because $\delta(H \cup C_2, a) = \delta(\{q\}, a) = \{q\}$.

If $C_1=\{q\}$ and $C_2=\emptyset$, then at the $\accfincof{1}$-leaf it obtains $S \cup C_1 = \{r,q\}$.
Since $q \xrightarrow{a} q$ is $\tcacc{1}{1}$-colored, it returns $\emptyset$, and this branch yields no successor.
\end{example}


We say that a tree-macrostate is \emph{satisfied} if all its $\Inf$-leaves have empty breakpoints. Formally, we define $\IsSat$ inductively over subtrees of $\psi$ by $\IsSat(\mstate_{\varphi_1 \star \varphi_2}) = \IsSat(\mstate_{\varphi_1}) \wedge \IsSat(\mstate_{\varphi_2})$ with $\star \in \{\wedge, \vee\}$, $\IsSat(M_{\accfincmgof{5}{\ell}}) = \top$, and $\IsSat(M_{\accinfcmgof{4}{k}}) \iff B = \emptyset$ where $B$ is the breakpoint set of $M_{\accinfcmgof{4}{k}}$.

\begin{figure}[t]
\begin{algorithm}[H]
  \caption{Successor computation $\transcompl((N, C, \mstate_\psi), a)$}
  \label{alg:induc-succ}

  $N' = \delta(N,a) \cap \states_n$\;
  \If{$C = \emptyset \wedge \IsSat(\mstate_{\psi})$}{
    $\mathcal{R} := \{ (N', \delta(N, a) \cap \states_d,\; \mstate_\psi') \mid \mstate_\psi' \in \GetSucc(\mstate_\psi, a, \emptyset, \bot) \}$\;
  } \uElse {
    $C' := \delta(C, a)$\;
    $\mathcal{R} := \{ (N', C',\; \mstate_\psi') \mid \mstate_\psi' \in \GetSucc(\mstate_\psi, a, \emptyset, \bot) \}$\;\label{line:checking}
    $\mathcal{R} := \mathcal{R} \cup \{ (N', \emptyset,\; \mstate_\psi') \mid \mstate_\psi' \in \GetSucc(\mstate_\psi, a, C, \top) \}$\;\label{line:sampling}
  } 
\Return $\mathcal{R}$
\end{algorithm}
\vspace*{-5mm}
\end{figure}

Now we have all the ingredients to construct $\transcompl$ and $\acccondcompl$.
Intuitively, $\autcompl$ operates in two alternating phases that repeat infinitely often:
First, in the \emph{sampling} mode, the construction samples runs entering $\states_d$ by maintaining a set $C$.
This mode must repeat infinitely often to ensure that every run is eventually tracked and checked against $\psi$.
Second, in the \emph{checking} mode, the sampled runs are verified to ensure that all of them satisfy $\psi$.
Formally, $\transcompl$ is defined by Algorithm~\ref{alg:induc-succ}.
Its behavior depends on the value of $C$.
When $C = \emptyset$ and
$\IsSat(\mstate_\psi)$ (i.e., checking mode has been successfully completed), an accepting mark is
emitted and $C$ is replenished with all states from $N \cap \states_d$ to begin a fresh sampling
cycle. 
Otherwise, the automaton is in the checking phase, meaning that some breakpoint set is not yet empty. In this case, we proceed nondeterministically: either the successors are computed without resampling, i.e., $f=\bot$, thereby continuing the verification of the currently tracked runs (see Line~\ref{line:checking}), or we choose to start a new sampling cycle by initiating resampling (see Line~\ref{line:sampling}), which nondeterministically assigns all current runs in $C$ into leaf nodes (see Line~\ref{line:sampling}).
Finally, the accepting condition of $\autcompl$ is $\acccondcompl = \{ (N, C, \mstate_\psi) \ltr{a} (N', C', \mstate_\psi') \in \transcompl \mid C = \emptyset, \IsSat(\mstate_\psi) \}$.
Let $\psi^{\text{DNF}}$ be the disjunctive normal form of $\psi$.
Then $\acccondcompl$ ensures that each run of the input word either never enters the deterministic component or, after entering, satisfies some conjunct of $\psi^{\text{DNF}}$ by visiting infinitely colours of all $\Inf$-leaves while avoiding colours of all $\Fin$-leaves.
It follows that $\langof \autcompl \subseteq \Sigma^\omega\setminus \langof  \aut$.
Conversely, for every word not accepted by $\aut$, there exists a correct nondeterministic choice for runs expected to satisfy a disjunct in $\psi^{\text{DNF}}$.
Along such a macrorun, the breakpoint sets associated with all $\Inf$-leaves are cleared infinitely often, while no $\Fin$-conditions are violated.
Hence, $\Sigma^\omega\setminus \langof \aut \subseteq \langof \autcompl$.

Therefore, we obtain our main result:
\begin{theorem}\label{thm:complement-elsa}
$\langof \autcompl = \Sigma^\omega\setminus \langof \aut$.
\end{theorem}

\vspace{-0.0mm}
\subsection{Tree-Macrostate Normalization}\label{}
\vspace{-0.0mm}
Two runs may merge along a word. Consequently, after computing successors over a letter, a successor tree-macrostate may contain redundant run tracking, in particular when the same state appears in multiple subtrees of a $\vee$-node and is checked independently. To avoid having multiple representations of semantically the same state, we keep only the rightmost occurrence and enforce disjointness among subtrees of each $\vee$-node. Intuitively, once two runs merge, one can be discarded, as its finite prefix does not affect the satisfaction of any acceptance condition. This normalization preserves correctness. See \trOrAppendix{sec:tree-reduction} for details.

\vspace{-0.0mm}
\subsection{Shared Breakpoint Optimization}\label{sec:shared-breakpoint}
\vspace{-0.0mm}

In the base construction, each $\accinf$-leaf maintains its own breakpoint, leading to a state-space factor of $|2^{\states_d}|^k$ for a subtree with $k$ $\Inf$-leaves. We can apply a standard round-robin technique to track one $\Inf$-leaf at a time using a shared breakpoint, reducing the factor to $k \cdot |2^{\states_d}|$. Concretely, the construction cycles through the $\Inf$-leaves, verifying each leaf in turn before proceeding to the next, and restarting the cycle thereafter. This optimization preserves correctness while avoiding independent breakpoint tracking.

To implement this, we attach a context $\mathit{ctx} = (h, j, \lst, B_{\mathit{sh}})$ to selected nodes, which maintain the shared breakpoint. The phase $h \in \{\ctxW, \ctxR, \ctxP\}$ indicates whether the cycle is waiting for the next sampling, resampling the shared breakpoint for the current leaf, or processing the current breakpoint. The sequence $\lst = \langle u_1,\dots,u_k \rangle$ orders the $\Inf$-leaves, $j$ denotes the currently tracked leaf, and $B_{\mathit{sh}} \subseteq \states_d$ is the shared breakpoint. Individual breakpoints are removed under this optimization. A tree-macrostate may contain multiple such contexts; details are given in \trOrAppendix{app:shared-breakpoint}.

\vspace{-0.0mm}
\subsection{OR-FIN Optimization}\label{sec:or-fin-opt}

In the base construction, $\GetSucc$ for a $\vee$-node nondeterministically partitions the states in $C$ across its children, yielding $2^{|C|}$ successors. The \emph{OR-FIN optimization} avoids this blowup when the left child $\varphi_1$ is an $\IsORF$ subtree, i.e., a $\vee$-composition of $\accfin$-leaves. The key idea is to propagate violating states sequentially: a $\accfin$-leaf $\accfin$ filters violating states/runs and passes them to the next $\accfin$-leaf. These states/runs are eventually either absorbed by an $\accinf$-leaf or cause rejection if they reach the root.
This is correct because a run needs to satisfy at least one branch of the $\vee$-subtree; thus, it suffices to check the $\accfin$-leaves one after another. If a run violates all of them, it is correctly rejected; otherwise, it is accepted as soon as one condition is met. This eliminates the need for exponential branching while preserving correctness.
Details of the construction are provided in \trOrAppendix{app:or-fin}.

\newcommand{\tabDetailedComplexity}[0]{
\begin{table}[t]
	\caption{Detailed complexity analysis}\label{table:detailed:complexity}
	\scalebox{.9}{
		\begin{tabular}{lcc}
			\textbf{Name} & \textbf{Classic} & \textbf{Generalized}
			\\\midrule
			\emph{\buchi} & \small $\accinfof{\tacc 0}$ & \small $\bigwedge_{j=0}^{k-1} \accinfof{\taccj}$
			\\
			&$x=1$, $y=0$, $d=0$, $m=0$ & $x=k$, $y=0$, $d=0$, $m=0$ 
			\\
			& $3 \cdot 5^n$ 
			& $3 \cdot (2k + 3)^n$ 
			\\\midrule
			\emph{co-\buchi} \qquad & \small $\accfinof{\tacc 0}$ & \small $\bigvee_{j=0}^{k-1} \accfinof{\taccj}$
			\\
			& $x=0$, $y=1$, $d=1$, $m=1$ & $x=0$, $y=k$, $d=k$, $m=1$ 
			\\
			& $3 \cdot 9^n$ 
			& $3 \cdot 2^{(k+3)n}$ 
			\\\midrule
			Rabin pair & \small $\accfinof{B} \land \accinfof{G}$ & \small $\accfinof{B} \land \bigwedge_{i=0}^{\ell -1}\accinfof{G_i}$
			\\
			& $x=1$, $y=1$, $d=1$, $m=1$ & $x=\ell$, $y=1$, $d=1$, $m=1$ 
			\\
			& $3 \cdot 13^{n}$ 
			& $3 \cdot (4\ell + 9)^{n}$ 
			\\\midrule
			\emph{Rabin} & \small \qquad$\bigvee_{j=0}^{k-1} \accfinof{B_j} \land \accinfof{G_j}$ \qquad & \small \qquad$\bigvee_{j=0}^{k-1} (\accfinof{B_j} \land \bigwedge_{i=0}^{\ell -1}\accinfof{G_{j,i}})$
			\\
			& $x=k$, $y=k$, $d=1$, $m=k$ & $x=k\ell$, $y=k$, $d=1$, $m=k$ 
			\\
			& $3 \cdot 2^{(3k+2)n}$ 
			& $3 \cdot 2^{(3k\ell+2)n}$ 
			\\\midrule
			Streett pair & \small $\accinfof{G} \lor \accfinof{B} $ & \small $\accinfof{G} \lor \bigvee_{i=0}^{\ell -1}\accfinof{B_i}$
			\\
			& $x=1$, $y=1$, $d=1$, $m=1$ & $x=1$, $y=\ell$, $d=\ell$, $m=1$ 
			\\
			& $3 \cdot 17^{n}$ 
			& $3 \cdot 2^{(\ell+4)n}$ 
			\\\midrule
			\emph{Streett} & \small $\bigwedge_{j=0}^{k-1} \accinfof{G_j}\lor \accfinof{B_j}$ & \small $\bigwedge_{j=0}^{k-1} (\accinfof{G_j} \lor \bigvee_{i=0}^{\ell-1}\accfinof{B_{j,i}})$
			\\
			& $x=k$, $y=k$, $d=1$, $m=k$ & $x=k$, $y=k\ell$, $d=\ell$, $m=k$ 
			\\
			& $6 \cdot k^{n} \cdot 2^{(2k+2)n}$ 
			& $6 \cdot k^{n} \cdot 2^{(2k\ell+2)n}$ 
			\\\midrule
			\emph{Parity} & \multicolumn{2}{c}{\small $\accfinof{\tacc 0} \land (\accinfof{\tacc 1} \lor (\accfinof{\tacc 2} \land (\accinfof{\tacc 3} \lor (\accfinof{\tacc 4} \land \ldots \accinfof{\taccgof{k}}) \ldots ))))$}
			\\
			& \multicolumn{2}{c}{\small $(\accfinof{\taccgof{0}} \land \accinfof{\taccgof{1}}) \lor (\accfinof{\taccgof{0,2}} \land \accinfof{\taccgof{1,3}}) \lor \ldots \lor (\accfinof{\taccgof{0, \ldots, k-1}} \land \accinfof{\taccgof{1, \ldots, k}})$}
			\\
			& \multicolumn{2}{c}{$x=\frac{k}{2}$, $y=\frac{k}{2}$, $d=1$, $m=\frac{k}{2}$}
			\\
			& \multicolumn{2}{c}{$3 \cdot 2^{(\frac{3}{2}k+2)n}$} 
		\end{tabular}
	}
\end{table}
}

\vspace{-0.0mm}
\subsection{Complexity Analysis}\label{sec:complexity}

Let us now analyse the state complexity of our construction.
Consider an ELSA~$\aut$ with $n$ states $\states = \states_n \cup \states_d$ and an Emerson-Lei condition
$\acccond$, together with its complement $\autcompl$ constructed by the
inductive procedure above.
Let $\psi = \lnot \acccond$ with $x$ being the number of $\Fin$ atoms, and $y$
being the number of $\Inf$ atoms.
We define $d$ as the maximum number of $\Inf$-leaves in a~$\land$-only subtree
of $\psi$, and $m$ as the number of $\land$-only maximal subtrees (i.e., $\land$-only subtrees not contained inside larger $\land$-only subtrees) of~$\psi$ with an $\Inf$-leaf.
Moreover, $m \leq y$ refers to the number of contexts generated by the shared breakpoint optimization (\cref{sec:shared-breakpoint}), and $d$ is the maximum number of leaves that contexts may schedule.
Observe that $x$, $y$, $m$, and $d$ are upper-bounded by the number of atoms of~$\acccond$.

The states of~$\autcompl$ hold a set $N \subseteq \states_n$, a set $C \subseteq \states_d$, and a tree-macrostate $\mstate_{\psi}$.
Without optimization, $\mstate_{\psi}$ holds $x$ sets $S \subseteq \states_d$ (one per $\Fin$-leaf), and $y$ pairs of sets $(S, B) \subseteq \states_d^2$ (one per $\Inf$-leaf).
With the shared breakpoint optimization, $\mstate_{\psi}$ holds $x+y$ sets $S \subseteq \states_d$ (one per leaf), and $m$ contexts.
Each context keeps a phase status $p \in \{W, R, P\}$, a counter $j \in \{1, \ldots, d\}$, a set of at most $d$ identifiers among $\Inf$-leaves $\lst$, and a breakpoint $B_{\mathit{sh}}\subseteq  \states_d$.
Note that $\lst$ is fully determined by the formula, and thus it is independent of the size of $\autcompl$.

We establish an upper bound on the number of states of $\autcompl$ by determining, for each state of $\aut$, the internal subsets of a macrostate in which it may appear.
This analysis assumes the uses of the shared breakpoint optimization.
Given a state $q\in \states$ of $\aut$, the possibilities are that 
(I)~$q$ belongs in $N$ and no other subsets,
(II)~$q$ belongs in $N$ and some other subsets:
(II.i)~$q$ belongs in $C$ or not, 
(II.ii)~for each $\Fin$-leaf holding the set $S \subseteq \states_d$, $q$~belongs in~$S$ or not,
(II.iii)~for each $\Inf$-leaf holding the pair of sets $(S, B) \subseteq \states_d^2$, $q$~belongs in~$S$ or not, and
(II.iv)~for each node holding the shared context $B_{\mathit{sh}}$, $q$~belongs in~$S$ or not.
Item~(I) identifies a single isolated case, all other cases are captured by the subitems of~(II).
Item~(II.i) describes two subcases, i.e., possibilities that stacks with the other subitems of~(II).
Item~(II.ii) considers all subsets of $\Fin$-leaves, thus describing $2^x$ subcases.
Item~(II.iii) is analogous for $\Inf$-leaves, thus describing $2^y$ subcases.
Finally, Item~(II.iv) is analogous for roots of $\land$-only maximal subtrees of~$\psi$ with a $\Inf$-leaf, thus describing $2^m$ subcases.
Observe that case where $q$ belongs in no subsets (including $N$) is captured by~(II).
Since this subset appearance is independent for all states of $\aut$, we have $(1+ 2 \cdot 2^{y} \cdot 2^{x} \cdot 2^{m})^{n}$ cases so far.
To get the full picture, it remains to include the data carried by each context, namely, the phase and the counter.
We take this information into account through copies of the previous reasoning.
Hence, $\autcompl$ admits at most
%
$$
3 \cdot \max\{d,1\} \cdot  (1 + 2 \cdot 2^{y} \cdot 2^{x} \cdot 2^{m})^{n}
$$
%
macrostates.
Indeed, there are three phase statuses and a counter that takes at most $d$ distinct values.
We derive upper bounds for Rabin, generalized Rabin, and Parity acceptance conditions from this general analysis.
We can further reduce this bound when a~certain property holds.
First, we assume that sets appearing in $\Fin$-leaves are disjoint.
This situation is fulfilled when all $\Fin$-leaves of~$\psi$ are in a subtree without $\land$-nodes.
So, this assumption holds true for $\psi$ when the acceptance condition of $\aut$ is generalized co-\buchi, Streett pair, generalized Streett pair, Streett, and generalized Streett, for instance.
We get $3 \cdot \max\{d,1\} \cdot (1 + 2 \cdot 2^{y} \cdot (x+1) \cdot 2^{m})^{n}$ 
because Item~(II.ii) now captures the subcases where $q$ appears in one of the $\Fin$-leaves, or none of them.
Second, we assume that sets appearing in $\Fin$-leaves are disjoint and additionally the whole negated acceptance condition $\psi$ has no conjunction.
Thus, the shared breakpoint optimization has no effect here and breakpoint sets appear on each $\Inf$-leaf.
In particular, $0 \leq d \leq 1$ ($0$ if $\psi$ has no $\Inf$-leaf, $1$ otherwise) and $m=y$.
This assumption holds true for $\psi$ when the acceptance condition of $\aut$ is \buchi, generalized \buchi, co-\buchi, Rabin pair, and generalized Rabin pair. 
We get $3 \cdot \max\{d,1\} \cdot (1 + 2 \cdot 2^{y} \cdot (x + m + 1))^{n}$
because Items~(II.ii) and~(II.iv) now capture together the subcases where $q$
appears in one of the leaves (since $m=y$), or none of them.

\tabDetailedComplexity

Upper bounds for popular acceptance conditions are given in
\cref{table:detailed:complexity} (a simplified version was already given in 
\cref{table:compare}).
We note that we could obtain better bounds for the concrete constructions in
\cref{sec:rabin,sec:streett} because of more fine-tuned construction and analysis.

\vspace{-0.0mm}
\section{Complementation of an ELEA and a General ELA}\label{sec:elevator-compl}
\vspace{-0.0mm}

\subparagraph*{Modular construction of an ELEA.}
To lift the complementation procedure from semi-deterministic to elevator automata,
we employ the modular complementation approach~\cite{HavlenaLLST23} generalized to the ELA
setting. The modular construction decomposes complementation of the whole automaton $\aut$ into
complementation algorithms for individual SCCs or groups of SCCs called partitions.
Each partial complementation algorithm specifies its own macrostates, transition function,
coloring, and acceptance condition. The transition function takes into account runs entering the
partition from outside. The inductive construction from \Cref{sec:inductive-compl} fits
naturally into this framework: the macrostates of the partial algorithm take the form
$(C,\mstate_\psi)$, where the $N$ component is omitted since incoming runs are tracked by
the top-level modular algorithm. Concretely, in \Cref{alg:induc-succ} the set $N'$ represents
the newly incoming runs supplied by the top-level algorithm, and $\states_d$ denotes the states
of the deterministic partition. To obtain a complete complementation algorithm for ELEAs, we
additionally plug the Miyano-Hayashi algorithm~\cite{MiyanoH84} into the modular construction
to handle inherently-weak partitions and the same algorithm with round robin
(similar to the procedure in \cref{sec:streett}) to handle generalized
co-\buchi partitions with the complexity~$\bigOof{3^n}$ and
$\bigOof{k\cdot(3^n)}$ respectively.
An example of the modular construction is shown in \trOrAppendix{app:modular-example}.


\subparagraph*{Elevatorization of ELAs.}
Any ELA can be transformed into an equivalent ELEA, which lifts our techniques to the general case.
The conversion treats each SCC independently.
In particular, nondeterministic accepting components (NACs) must be
semi-determinized~\cite{CourcoubetisY88} since they violate the elevator
condition.
Because semi-determinization yields only inherently-weak and deterministic components, the resulting automaton is an ELEA.
Let $\aut$ be an ELA with acceptance condition $\acccond$.
For each NAC, $\acccond$ can be simplified to contain only colors appearing in the NAC.
This yields a smaller local condition that is used in place of $\acccond$ when applying the semi-determinization construction of~\cite{JohnJBK21} to the corresponding NAC.
After semi-determinization, a copy of the NAC without colors is connected to a~deterministic part that emits a fresh color $\tcacc{4}{$\top$}$, and $\Infof{\tcacc{4}{$\top$}}$ becomes the acceptance condition of this semi-deterministic component.
To prevent interference between SCCs of $\aut$, we disable $\acccond$ within all semi-determinized components by marking them with a fresh color $\tcacc{5}{$\bot$}$, and adopt $(\acccond \land \Finof{\tcacc{5}{$\bot$}}) \lor \Infof{\tcacc{4}{$\top$}}$ as the global acceptance condition.
The semi-determinization of~\cite{JohnJBK21} runs in $\bigO(2^{|\varphi|+|Q|})$, where~$Q$ is the set of states and $\varphi$ is the simplified acceptance condition of the given SCC.
Details are given in~\trOrAppendix{app:elevatorization}.

\definecolor{benchred}{HTML}{E31A1C}
\definecolor{benchblue}{HTML}{1F78B4}
\definecolor{benchgreen}{HTML}{33A02C}
\newcommand{\figScatterCompl}[0]{
	\begin{figure}[b]
		\centering
		\begin{subfigure}{0.32\textwidth}
			\centering
			\includegraphics[width=4cm]{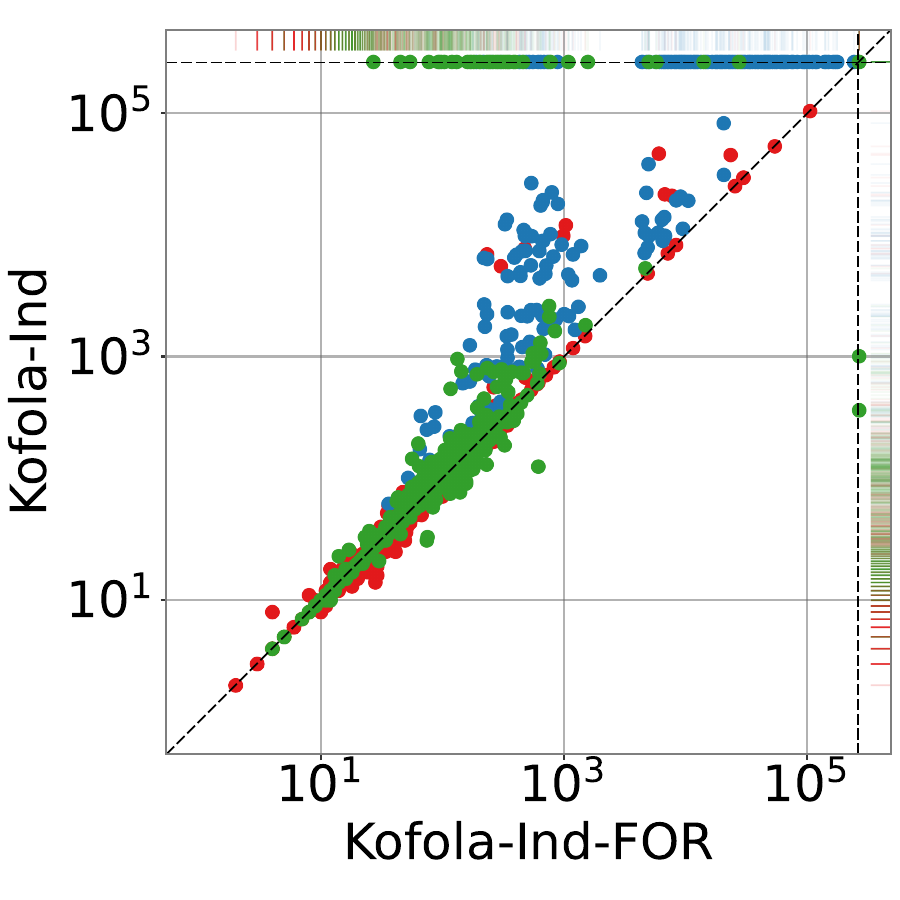}
      \vspace*{-2mm}
      \captionsetup{justification=centering}
      \caption{}\label{fig:scatter-ind-for}
		\end{subfigure}
		\begin{subfigure}{0.32\textwidth}
			\centering
			\includegraphics[width=4cm]{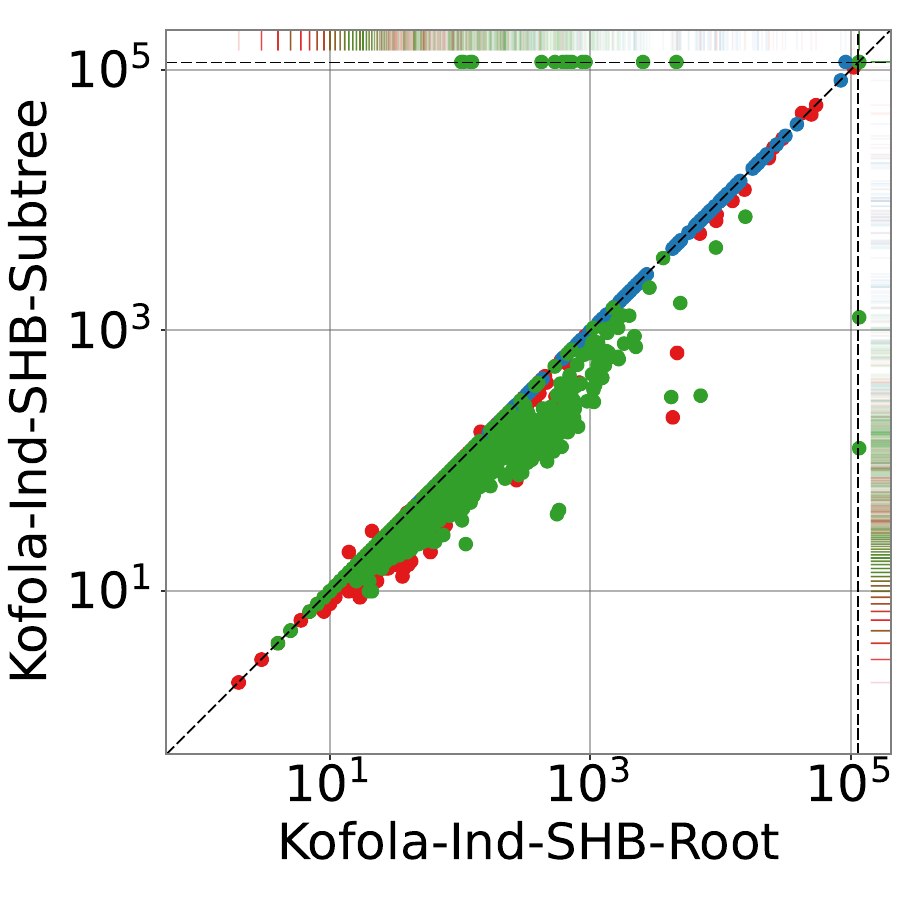}
      \captionsetup{justification=centering}
      \vspace*{-2mm}
      \caption{}\label{fig:scatter-shb}
		\end{subfigure}
		\begin{subfigure}{0.32\textwidth}
			\centering
			\includegraphics[width=4cm]{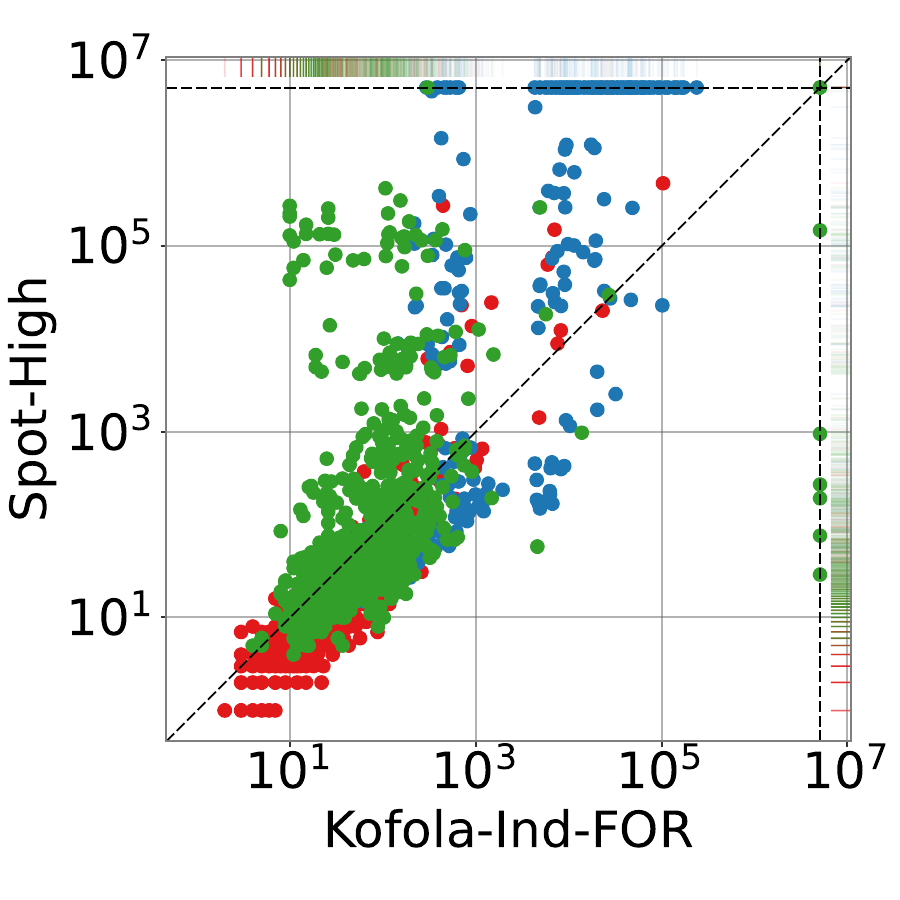}
      \vspace*{-2mm}
      \captionsetup{justification=centering}
      \caption{}\label{fig:scatter-spot}
		\end{subfigure}
    \vspace*{-2mm}
		\caption{
				Scatter-plots comparing sizes of complemented automata (in the number of states). 
				Colors denote benchmarks:
				\textcolor{benchblue}{\(\bullet\)} \textbf{GRA},
				\textcolor{benchgreen}{\(\bullet\)} \textbf{LTL-Rand}, and
				\textcolor{benchred}{\(\bullet\)} \textbf{LTL-Lit}.
				}\label{fig:scatter-compl}
  \vspace*{-4mm}
	\end{figure}
}

\newcommand{\tabCompl}{
\begin{table}[t]
\newcolumntype{h}{>{\columncolor{Gray!30}}r}
\newcolumntype{g}{>{\columncolor{Gray!30}}c}
\renewcommand{\arraystretch}{1.15}
\centering
\caption{Statistics for complementation. The column \textbf{unsolved} shows the number of
automata for which complementation failed (out of 2,693), including timeouts and errors such as running
out of resources
or exceeding the maximum number of acceptance sets. The columns \textbf{average},
\textbf{median}, and \textbf{total states} give the average, median, and total numbers of states of the resulting automata (across successful runs). The column \textbf{time}
contains the total runtime (in seconds).}
\vspace{-2mm}
\scalebox{0.8}{%
\label{tab:compl}
\begin{tabular}{l h r h r h r}
\toprule
\multicolumn{1}{c}{\bf tool} &
\multicolumn{1}{g}{\bf unsolved} &
\multicolumn{1}{c}{\bf average} &
\multicolumn{1}{g}{\bf median} &
\multicolumn{1}{c}{\bf total states} &
\multicolumn{1}{g}{\bf time} \\
\midrule
\rowcolor{GreenYellow}\kofolafor  & 40  & 2,436  & 25 & 6,463,028 & 2,828 \\
\kofolaind      & 247  & 537   & 21 & 1,312,689 & 3,629 \\
\kofolashbroot  & 226  & 639    & 26 & 1,576,427 & 3,502 \\
\kofolashbsubtree  & 242  & 547    & 23 & 1,339,994 & 3,066 \\
\spot           & 156 & 15,058 & 39 & 38,202,295  & 673  \\
\spothigh      & 159 & 11,136 & 15 & 28,218,675 & 1,167 \\
\bottomrule
\end{tabular}}
\vspace*{-3mm}
\end{table}
}

\vspace{-0.0mm}
\section{Experimental Evaluation}\label{sec:experiments}
\vspace{-0.0mm}

\subparagraph*{Implementation. }
We extended \kofola~\cite{AlexajHHLLM26} with direct support for ELA complementation and integrated the
inductive construction into its modular framework. We evaluate four configurations of the inductive construction:
(i)~\kofolaind: the base inductive construction without any optimizations;
(ii)~\kofolashbroot: includes the shared breakpoint optimization, where a~single context is maintained at the root of each tree-macrostate;
(iii)~\kofolashbsubtree: includes the shared breakpoint optimization, where a~context is maintained within each subtree consisting solely of $\land$-nodes and $\accinf$-leaves (allowing multiple contexts per tree-macrostate); and
(iv)~\kofolafor: employs the OR-FIN optimization.

\subparagraph*{Used tools and environment.}
We experimentally evaluated the inductive construction and compared it with the
state-of-the-art tool \spot~\cite{Duret-LutzRCRAS22} (we use a~custom build
with the maximum number of colors increased from 32 to 128) and its variant
\spothigh, which applied \spot's \textsc{High} reduction on the output.
All experiments were conducted on an Ubuntu GNU/Linux 24.04 virtual machine
with 64\,GiB of RAM, running on an AMD EPYC 9124 CPU server under the Proxmox
Virtual Environment. The timeout was 120\,s.
The correctness of results was validated by cross-checking the outputs using
\spot's \texttt{autcross}.

\subparagraph*{Benchmarks.}
Our evaluation is based on the following three data sets (available
at~\cite{automata-benchmarks}):
\begin{itemize}
    \item \textbf{Existing LTL benchmarks.}
    We used a collection of LTL formulas from the literature, taken from the benchmarks used in the evaluation of \texttt{ltl3tela}~\cite{ltl3tela}.
	The benchmark set is available in the \texttt{ltl3tela} repository~\cite{ltl3tela-benchmarks} and is based on formulas generated by \texttt{genltl}
	as well as formulas collected from prior work. After filtering out trivial cases, this resulted in 1,378 instances. These automata were then transformed into elevator
	automata. We denote this benchmark as \textbf{LTL-Lit}. 
    
    \item \textbf{Randomly generated automata.}
    We generated 315 automata with a single generalized Rabin pair acceptance and transformed them into elevator automata while preserving the generalized Rabin pair acceptance form. We denote this 
	benchmark as \textbf{GRA}. 
    
    

    \item \textbf{Random LTL formulas.}
    We used \spot's \texttt{randltl} with parameters \verb@--weak-fairness@
    \verb@-l --seed=42128971 -n -1 --simplify=0 a b c d e@ to generate
    random LTL formulas.
    The generated formulas were subsequently translated into automata in the \texttt{HOA} format using \texttt{ltl3tela}~\cite{ltl3tela}. 
    We used \texttt{autfilt} to remove trivial automata
	(such as deterministic, weak, empty, or terminal ones), and kept only elevator automata.
	This way, we obtained 1,000 instances. We denote this benchmark as \textbf{LTL-Rand}.
\end{itemize}

\figScatterCompl  

\tabCompl  

\vspace{-5mm}
\subparagraph*{Discussion.}
%

We give the results comparing sizes of outputs of our four configurations
and \spot in \cref{fig:scatter-compl,tab:compl}.
\Cref{fig:scatter-ind-for} evaluates the OR-FIN optimization
(\cref{sec:or-fin-opt}) against the baseline; it~shows that the optimization is
mostly advantageous and essential for solving some benchmarks (it solves 207
more benchmarks than the baseline; we note that in a few cases, the baseline,
however, could solve some benchmarks that the optimization could not).
The second plot, in \cref{fig:scatter-shb}, evaluates the two shared breakpoint
optimizations; it~shows that although \kofolashbroot solves slightly more benchmarks,
\kofolashbsubtree almost always outputs a slightly smaller automaton
(if it finishes).
The last plot, \cref{fig:scatter-spot}, evaluates our best configuration,
\kofolafor, against \spothigh;
on a part of the inputs, the two tools give comparable results (\spothigh in
many cases giving slightly smaller outputs), however, on more challenging
inputs, \kofolafor mostly dominates \spothigh by a~large margin and can solve
119 more benchmarks.

In \cref{tab:compl}, we provide statistical data for our experiments.
We highlight \kofolafor, which solved the most benchmarks by a~large margin
(the second one was \spot with 116 less solved benchmarks).
Although the average size of the output is larger than other configurations of
\kofola, this is mostly because \kofolafor can complement many of the more
challenging inputs (this also holds for the median and the total number of states).
On the other hand, even with the more finished runs, \kofolafor was the fastest
configuration of \kofola (but still not as good as \spot---which is quite
optimized---when we compare the runtime on only finished instances).
We emphasize the total number of states generated by \kofolafor: roughly
$4.4\times$ less than \spothigh despite solving many more cases.

\vspace{-0.0mm}
\section{Related Work}
\vspace{-0.0mm}

Complementation of finite automata on infinite words is a~problem addressed
since the trailblazing work of \buchi~\cite{buchi1962decision}. 
For the basic model of \buchi automata, the lower bound for the problem is
$\Omega((0.76n)^n)$ established using the full automata technique
in~\cite{Yan06} by Yan (improving the previous $\Omega(n!)$ lower bound of
Michel~\cite{michel1988complementation}).
A~multitude of approaches for complementing BAs have been proposed in the literature, 
e.g.,
Ramsey-based~\cite{BreuersLO12,buchi1962decision,sistla1987complementation},
rank-based~\cite{HavlenaL21,HavlenaLS22a,HavlenaLS22b,ChenHL19,Vardi07,KupfermanV01,Schewe09}, 
determinization-based~\cite{Safra88,Piterman07,LiTFVZ22},
slice-based~\cite{KahlerW08}, and others~\cite{AllredU18,HavlenaLLST23,LiTZS18}.
Some of these approaches~\cite{Schewe09,AllredU18} even match the lower bound
(modulo a~small polynomial factor).
More efficient complementation algorithms for structural subclasses of BAs were
also considered, e.g., for 
inherently-weak ($\frac 2 3 3^n$)~\cite{MiyanoH84},
deterministic ($n+1$ with output being a co-\buchi automaton or $2n$ with
output being a~BA)~\cite{Kurshan87},
semideterministic ($4^n$)~\cite{BlahoudekHSST16,ChenHLLTTZ18},
elevator (a GBA with $4^n$ states and two colors)~\cite{HavlenaLS22a,HavlenaLLST23}, or
unambiguous ($4^n$)~\cite{LiVZ20,FengLTVZ23} BAs.

The landscape is considerably less explored for richer acceptance conditions
(we recommend Boker's excellent survey~\cite{Boker18} and its accompanying web
page for more details and references).
For generalized BAs (GBAs) with~$n$ states and~$k$ acceptance sets,
\cite{KupfermanV05}~gives a~complementation algorithm with the
complexity~$2^{\bigOof{n \log kn}}$.
The GBA complementation algorithm given in~\cite{HavlenaLS25} produces
a~result as a~BA with $\bigOof{n(0.76nk)^n}$ states and~\cite{Dokoupil26}
claims to improve over this.
Co-\buchi automata can be complemented into BAs with $\frac 2 3 3^n$
states~\cite{MiyanoH84,Boker18}.
For a Rabin automaton with~$n$ states and~$\ell$ pairs, the algorithm
in~\cite{KV05} produces a~complement BA with $\bigOof{\ell \cdot 3^n \cdot
(2n+1)^{n\ell}}$ states and the one in~\cite{HavlenaLS25} produces a~GBA with
$\bigOof{n^\ell(0.76n)^{n\ell}}$ states and~$\ell$ colours.
Streett automata can be complemented into BAs with~$2^{\bigOof{n\ell \log n \ell}}$
states using~\cite{KV05} and into BAs with $2^{\bigOof{n \log n + n\ell \log \ell}}$ states~\cite{CaiZ11b}.
Parity automata can be complemented into BAs with $2^{\bigOof{n \log n}}$ states.
In~\cite{SafraV89}, it is shown that the lower bound for the number of states
of a~complement of an ELA is~$2^{2^n}$ (even if the output can also be an ELA),
and the same work shows that the upper bound is~$2^{2^{\bigOof{n}}}$ (not
considering the complexity of the acceptance condition as a~parameter).
This was refined in~\cite{HavlenaLS25}, which gave an algorithm that produces
a~GBA with~$2^k$ colors and $\bigOof{n^{2^k}(0.76nk)^{n2^k}}$ states, where~$n$
is the number of states of the input ELA and $k$ is the number of colors that
occur in the acceptance condition.
All of our upper bounds are better.

Elevator automata were introduced in~\cite{HavlenaLS22a} in the context of
rank-based complementation as a~subclass of BAs with a more efficient
complementation procedure ($\bigOof{16^n}$).
Since then, it was shown that there exists on optimal determinization procedure
($\bigOof{n!}$)~\cite{LiTFVZ22} and an efficient complementation procedure
($\bigOof{4^n}$)~\cite{HavlenaLLST23} for them.
A recent paper~\cite{AlexajHHLLM26} shows that the vast majority of BAs
considered in practical scenarios are elevator automata.

\clearpage
\bibliographystyle{plainurl}
\bibliography{literature}

\ifTR
\clearpage
\appendix

\crefalias{section}{appendix}
\crefalias{subsection}{appendix}

\vspace{-0.0mm}
\section{Proof of \cref{thm:streett-size}}\label{sec:proof-thm-streett-size}
\vspace{-0.0mm}

We can deduce the bound on the number of states of~$\autcompl$ by checking, for
each state $q \in \states$, in which internal sets of a~macrostate $\calM =
(N,\crset, S, B, z)$ it can be and give each such a~choice
a~number~$f(q)$ as follows:
\begin{itemize}
  \item  If~$q$ is not present in~$\calM$ at all, we assign $f(q) = 1$.
  \item  We do a case split on the value of~$z$.
    For $z=0$ (tracking $\Infof{\tacc 0}$):
    \begin{itemize}
      \item  If $q \in \states_n \cap N$ or $q \in \states_d \cap \crset \cap
        B$, we assign $f(q) = 2$ (i.e., we can merge~$N$ and $\crset\cap B$
        similarly as for a~generalized Rabin pair).
      \item  If $q \in (\crset \cap S)$, we let $f(q) = 3$ and
        if $q \in \crset \setminus (S \cup B)$, we let $f(q) = 4$.
        We note that the case $\crset \cap S \cap B$ does not need to be
        considered since for~$z=0$, the sets~$S$ and~$B$ are disjoint.
    \end{itemize}
  \item  For $z >0$ (tracking $\Finof{\colj}$):
    \begin{itemize}
      \item  If $q \in \states_n \cap N$ or $q \in \states_d \cap \crset \cap
        B \cap S$, we assign $f(q) = 2$.

      \item  If $q \in (\crset \cap S) \setminus B$, we let $f(q) = 3$ and
        if $q \in \crset \setminus S$, we let $f(q) = 4$.
        We note that the case $\crset \cap (B \setminus S)$ does not need to
        be considered since for~$z>0$ we have $B \subseteq S$.
    \end{itemize}
\end{itemize}
The total number of functions $f\colon \states \to \{1, \ldots, 4\}$ is
$4^n$ where $n = |\states|$.
Since $z$ is between~$0$ and~$k$, we obtain $(k+1)\cdot 4^n$.

\section{Detailed Optimizations}

\vspace{-0.0mm}
\subsection{Tree-macrostate Reduction}\label{sec:tree-reduction}
\vspace{-0.0mm}

After computing the successor of a tree-macrostate, the resulting tree may contain runs that are checked by multiple leaves. In particular, the same state may appear in both branches of a $\vee$-node and thus be checked independently along each branch, although it suffices to check it on one branch only. The \emph{tree-macrostate reduction} eliminates this redundancy by enforcing disjointness at $\vee$-nodes without affecting correctness.

\subparagraph*{Tree-macrostate restriction.}
We first define a $\mathtt{gather}(\mstate_\varphi)$ function that collects all states referenced 
by a tree-macrostate:
\begin{align*}
  \mathtt{gather}(\mstate_{\varphi_1 \wedge \varphi_2}) &= \mathtt{gather}(\mstate_{\varphi_1}) \cup \mathtt{gather}(\mstate_{\varphi_2}) && \mathtt{gather}(S_{\accfincmgof{5}{\ell}}) = S \\
  \mathtt{gather}(\mstate_{\varphi_1 \vee \varphi_2}) &= \mathtt{gather}(\mstate_{\varphi_1}) \cup \mathtt{gather}(\mstate_{\varphi_2}) && \mathtt{gather}((H,B)_{\accinfcmgof{4}{k}}) = H
\end{align*}
Note that only the \emph{track} set is gathered from $\accinf$-leaves; the
breakpoint $B$ is not considered, as it is a subset of the track set.
Further, for a set of states $R\subseteq \states_d$ we define the tree-macrostate restriction $\restrict{\mstate_\varphi}{F}$
for removing $R$ from every leaf in the tree:
\begin{align*}
  \restrict{\mstate_{\varphi_1 \wedge \varphi_2}}{R} &= \restrict{\mstate_{\varphi_1}}{R} \wedge \restrict{\mstate_{\varphi_2}}{R} && \restrict{S_{\accfincmgof{5}{\ell}}}{R} = (S \setminus R)_{\accfincmgof{5}{\ell}} \\
  \restrict{\mstate_{\varphi_1 \vee \varphi_2}}{R} &= \restrict{\mstate_{\varphi_1}}{R} \vee \restrict{\mstate_{\varphi_2}}{R} && \restrict{(H,B)_{\accinfcmgof{4}{k}}}{R} = (H \setminus R,\; B \setminus R)_{\accinfcmgof{4}{k}}
\end{align*}
For $\accinf$-leaves, both the track set and the breakpoint are restricted, since
the breakpoint is always a subset of the track set and must remain consistent with it.

\subparagraph*{Reduction.}
The reduction $\mathtt{clip}(\mstate_\varphi)$ is defined recursively over the
tree-macrostate structure:
\begin{align*}
  \mathtt{clip}(\mstate_{\varphi_1 \wedge \varphi_2}) &=
    \mathtt{clip}(\mstate_{\varphi_1}) \wedge \mathtt{clip}(\mstate_{\varphi_2}) && \mathtt{clip}(S_{\accfincmgof{5}{\ell}}) = S_{\accfincmgof{5}{\ell}} \\
  \mathtt{clip}(\mstate_{\varphi_1 \vee \varphi_2}) &=
    \mathtt{clip}\!\left(\restrict{\mstate_{\varphi_1}}{\mathtt{gather}(\mstate_{\varphi_2}')}\right) \vee \mstate_{\varphi_2}' && \mathtt{clip}((H,B)_{\accinfcmgof{4}{k}}) = (H,B)_{\accinfcmgof{4}{k}}
\end{align*}
where $\mstate_{\varphi_2}' = \mathtt{clip}(\mstate_{\varphi_2})$.
Leaf nodes are returned unchanged. For $\wedge$-nodes, both children are reduced
independently. For a $\vee$-node, the right child is reduced first; its states
$\mathtt{gather}(\mathtt{clip}(\mstate_{\varphi_2}))$ are then removed from the left child before
it to be reduced. This ensures that every state appears in at most one
branch of each $\vee$-node.

\subsection{Shared Breakpoint Optimization}\label{app:shared-breakpoint}

In order to reference a leaf node, we identify leaves with unique identifiers and 
we refer them as $\accinfidof{u}$-nodes ($\accinf$-node identified by $u$).
The context is assigned to $\land$-nodes that are roots of subtrees containing $\accinf$-leaves (we call such nodes \emph{owners} of a context).
Initially, an owner's context is $(\ctxW, 1, \langle u_1,\dots,u_k\rangle, \emptyset)$
where $u_1,\dots,u_k$ are the identifiers of all $\accinf$-leaves in its subtree.
When $\GetSucc$ is called on an owner, the context is first advanced by $\SuccCtx$ using the
resample flag $r$, and the updated context is then propagated into the subtree.
The procedure $\SuccCtx$ is given in \Cref{alg:succ-ctx}.

\begin{algorithm}[t]
  \caption{Context successor $\SuccCtx(\mathit{ctx}, r)$}
  \label{alg:succ-ctx}

  \lIf{$\mathit{ctx} = (\ctxW, j, \lst, B_{\mathit{sh}})$ and $r = \top$}{
    \Return $(\ctxR, 1, \lst, \emptyset)$
  }
  \lElseIf{$\mathit{ctx} = (\ctxW, j, \lst, B_{\mathit{sh}})$}{
    \Return $(\ctxW, j, \lst, B_{\mathit{sh}})$
  }
  \lElseIf{$\mathit{ctx} = (\ctxP, j, \lst, B_{\mathit{sh}})$ and $B_{\mathit{sh}} = \emptyset$ and $j < |\lst|$}{
    \Return $(\ctxR, j{+}1, \lst, \emptyset)$
  }
  \lElseIf{$\mathit{ctx} = (\ctxP, j, \lst, B_{\mathit{sh}})$ and $B_{\mathit{sh}} = \emptyset$ and $j = |\lst|$}{
    \Return $(\ctxW, 1, \lst, \emptyset)$
  }
  \lElse{
    \Return $\mathit{ctx}$
  }
\end{algorithm}

\subparagraph*{Tree successor.}

The function $\GetSuccSHB(M_\psi, a, C, f, \mathit{ctx})$ extends the base $\GetSucc$
with a context parameter $\mathit{ctx}$ and returns pairs $(m, \mathit{ctx}')$ of a
successor tree-macrostate and an updated context.
When $\mathit{ctx} = \nil$, the function immediately falls back to the base inductive
successor $\GetSucc(M_\psi, a, C, f)$ (cf.\ \Cref{alg:induc-succ-tree}) and returns
$\nil$ as the context.
For $\Fin$-leaves the successor is computed as in the base construction and the context
is passed through unchanged.
For an $\Inf$-leaf identified by $u$ with context $(q, j, \lst, B_{\mathit{sh}})$,
there are three cases:
(i) if $u \notin \lst$, the leaf is not tracked by this context and the base successor
is used with the context unchanged;
(ii) if $u = \lst_j$ and $q = \ctxR$, the leaf resamples the shared breakpoint to
$\delta(S \cup C, a)$ and transitions to phase $\ctxP$;
(iii) if $u = \lst_j$ and $q = \ctxP$, the leaf advances the shared breakpoint by
removing $\tcacc{4}{$k$}$-coloured transitions.
Note that tracked $\Inf$-leaves hold $(S, \emptyset)$ — their individual breakpoint is
replaced by the shared one — while untracked $\Inf$-leaves still carry their own $(S, B)$.
For $\wedge$-nodes that own the context, the context is first advanced by $\SuccCtx(ctx, f)$
before descending; the owner stores the updated context returned by whichever child
changed it (at most one child modifies the context per step).
For all other internal nodes ($\wedge$ non-owner and $\vee$), the context is forwarded
to both children and the updated context is selected by the same preference rule.
The complete procedure is given in \Cref{alg:shb-succ-tree}.

\begin{algorithm}[t]
  \caption{Shared-breakpoint successor $\GetSuccSHB(M_{\psi}, a, C, f, \mathit{ctx})$}
  \label{alg:shb-succ-tree}

  \If{$\mathit{ctx} = \nil$}{
    \Return $\GetSucc(M_{\psi}, a, C, f)\times\{ \nil \}$ \tcp*[r]{no context (\Cref{alg:induc-succ-tree})}
  }
  \If{$\psi = \accfingcof{5}{\ell}$ with $e = S_{\accfincmgof{5}{\ell}}$}{
    $L := \emptyset$ if $\exists\, t = s \ltr a s' \in \delta$ for $s \in S \cup C$ s.t.\ $\tcacc{5}{$\ell$} \in \colouring(t)$, otherwise $L := \{\delta(S \cup C, a)\}$\;
    \Return $\bigl\{ \bigl(\tree(\accfingcof{5}{\ell},\; L_{\accfincmgof{5}{\ell}},\; \nil,\; \nil\bigr),\; \mathit{ctx} ) \bigr\}$\;
  }
  \uElseIf{$\psi = \accinfgcof{4}{k}$ with $e = (S, B)_{\accinfidofsup{u}}$ and $\mathit{ctx} = (q, j, \lst, B_{\mathit{sh}})$}{
    $\mstate' := \tree\bigl(\accinfgcof{4}{k},\; (\delta(S \cup C, a), \emptyset)_{\accinfidofsup{u}},\; \nil,\; \nil\bigr)$\;
    \If{$u \notin \lst$}{
      \Return $\{ (\GetSucc((S,B)_{\accinfidofsup{u}}, a, C, f),\; \mathit{ctx}) \}$ \tcp*[r]{leaf not tracked by $\mathit{ctx}$}
    }
    \lElseIf{$u = \lst_j$ and $q = \ctxR$}{
      \Return $\bigl\{ \bigl(\mstate',\; (\ctxP,\; j,\; \lst,\; \delta(S \cup C, a))\bigr) \bigr\}$
    }
    \lElseIf{$u = \lst_j$ and $q = \ctxP$}{
      \Return $\bigl\{ \bigl(\mstate',\; (\ctxP,\; j,\; \lst,\; \delta^{\acccmgof{4}{k}}(B_{\mathit{sh}}, a))\bigr) \bigr\}$
    }
  }
  \uElseIf{$\psi = \varphi_1 \wedge \varphi_2$ and $M_\psi$ owns $\mathit{ctx}$}{
    $\mathit{ctx}' := \SuccCtx(\mathit{ctx}, f)$ \tcp*[r]{advance context before descending}
    \Return $\bigl\{ \bigl(\tree(\wedge,\; \mathit{ctx}'',\; m_1,\; m_2), \nil\bigr) \bigm| (m_1, \mathit{ctx}_1) \in \GetSuccSHB(M_{\varphi_1}, a, C, f, \mathit{ctx}'), (m_2, \mathit{ctx}_2) \in \GetSuccSHB(M_{\varphi_2}, a, C, f, \mathit{ctx}'), \mathit{ctx}'' = \mathit{ctx}_1$ if $\mathit{ctx}_1 \neq \mathit{ctx}'$ else $\mathit{ctx}_2 \bigr\}$\;
  }
  \uElseIf{$\psi = \varphi_1 \wedge \varphi_2$}{
    \Return $\bigl\{ \bigl(\tree(\wedge, -, m_1, m_2),\; \mathit{ctx}'' \bigr) \bigm| (m_1, \mathit{ctx}_1) \in \GetSuccSHB(M_{\varphi_1}, a, C, f, \mathit{ctx}), (m_2, \mathit{ctx}_2) \in \GetSuccSHB(M_{\varphi_2}, a, C, f, \mathit{ctx}), \mathit{ctx}'' = \mathit{ctx}_1$ if $\mathit{ctx}_1 \neq \mathit{ctx}$ else $\mathit{ctx}_2 \bigr\}$\;
  }
  \uElseIf{$\psi = \varphi_1 \vee \varphi_2$}{
    \Return $\bigl\{ \bigl(\tree(\vee, -, m_1, m_2),\; \mathit{ctx}'' \bigr) \bigm| (m_1, \mathit{ctx}_1) \in \GetSuccSHB(M_{\varphi_1}, a, C_1, f, \mathit{ctx}), (m_2, \mathit{ctx}_2) \in \GetSuccSHB(M_{\varphi_2}, a, C_2, f, \mathit{ctx}), C = C_1 \cup C_2, \mathit{ctx}'' = \mathit{ctx}_1$ if $\mathit{ctx}_1 \neq \mathit{ctx}$ else $\mathit{ctx}_2 \bigr\}$ \tcp*[r]{Split $C$ nondeterministically}
  }
\end{algorithm}

\subparagraph*{Satisfaction.}
Since Inf-leaves carry no individual breakpoints, a $\wedge$-node owning a context 
$\mathit{ctx} = (q,j,\lst,B_{\mathit{sh}})$ is satisfied exactly when the shared cycle has completed, 
i.e., the context has returned to wait with an empty breakpoint:
\[
  \IsSat\bigl(\mstate_{\varphi_1 \wedge \varphi_2},\, \mathit{ctx}\bigr) = (q = \ctxW \wedge B_{\mathit{sh}} = \emptyset)
  \wedge\IsSat(\mstate_{\varphi_1})\wedge \IsSat(\mstate_{\varphi_2}).
\]
For the rest of the nodes, the satisfiability remains the same.
Using the shared breakpoint optimization we are able to reduce the factor from $|2^Q|^k$ to $k|2^Q|$.

\subsection{OR-FIN Optimization}\label{app:or-fin}

The complete procedure for computing OR-FIN tree-macrostate successor is given in \Cref{alg:orf-succ-tree}. 
The procedure for computing the transition function $\transcompl$ is given as in \Cref{alg:induc-succ}
with a difference that if $\GetSuccORF$ returns nonempty violating states, this tree-macrostate 
is removed from the set. 

\begin{algorithm}[t]
  \caption{OR-FIN successor runs $\GetSuccORF(M_{\psi}, a, C, r)$}
  \label{alg:orf-succ-tree}

  \If{$\psi = \accfingcof{5}{\ell}$ with $e = S_{\accfincmgof{5}{\ell}}$}{
    $S' :=$ largest subset of $S \cup C$ s.t.\ no transition from $S'$ over $a$ is labelled by $\tcacc{5}{$\ell$}$\;
    $M := \bigl\{ \bigl(\tree(\accfingcof{5}{\ell},\; \delta(S', a)_{\accfincmgof{5}{\ell}},\; \nil,\; \nil\bigr),\; (S \cup C) \setminus S'\bigr) \bigr\}$\;
  } \uElseIf{$\psi = \accinfgcof{4}{k}$ with $e = (H, B)_{\accinfcmgof{4}{k}}$}{
    $L := (\trans(H\cup C, a), \trans(H\cup C, a))$ if $r = \top$, otherwise $L := (\trans(H, a),\, \delta(B,a) \setminus \{t \in \delta \mid \tcacc{4}{$k$} \in \colouring(t)\})$\;
    $M := \bigl\{ \bigl(\tree(\accinfgcof{4}{k},\; L_{\accinfcmgof{4}{k}},\; \nil,\; \nil\bigr),\; \emptyset\bigr) \bigr\}$\;
  } \uElseIf{$\psi = \varphi_1 \vee \varphi_2$ and $\IsORF(\varphi_1)$}{
    $M := \bigl\{ \bigl(\tree(\vee, -, m_1, m_2),\; V_2\bigr) \bigm|$
    $(m_1, V_1) \in \GetSuccORF(M_{\varphi_1}, a, C, r)$,
    $(m_2, V_2) \in \GetSuccORF(M_{\varphi_2}, a, V_1, [\![V_1 = \emptyset]\!]) \bigr\}$\;
  } \uElseIf{$\psi = \varphi_1 \vee \varphi_2$}{
    $M := \bigl\{ \bigl(\tree(\vee, -, m_1, m_2),\; v_1 \cup v_2\bigr) \bigm| (m_1, v_1) \in \GetSuccORF(M_{\varphi_1}, a, C_1, r),\; (m_2, v_2) \in \GetSuccORF(M_{\varphi_2}, a, C_2, r),\; C = C_1 \cup C_2\bigr\}$ \tcp*[r]{Split $C$ nondeterministically into $C_1$ and $C_2$}
  } \uElseIf{$\psi = \varphi_1 \wedge \varphi_2$}{
    $M := \bigl\{ \bigl(\tree(\wedge, -, m_1, m_2),\; v_1 \cup v_2\bigr) \bigm| (m_1, v_1) \in \GetSuccORF(M_{\varphi_1}, a, C, r),\; (m_2, v_2) \in \GetSuccORF(M_{\varphi_2}, a, C, r)\bigr\}$\;
  }
  \Return $M$
\end{algorithm}

\section{Example of the Modular Construction}\label{app:modular-example}

An example of the modular construction combining the procedure for the 
inherently-weak partition and the inductive construction for the deterministic 
partition is shown in \Cref{fig:modular-example}. 

\begin{figure}[ht]
  \begin{subfigure}[b]{0.1\linewidth}\centering
	\scalebox{0.5}{\begin{tikzpicture}[automaton,node distance=30mm]
  \node[state, initial] (p) {$p$};
  \node[state,above of=p] (q) {$q$};

  \path[->]
    (p) edge[loop right] node[above] {$a$} (p)
    (p) edge node[above] {$a$} (q)
    (q) edge[loop above] pic[auto] {acc=1} node[above] {$b$} (q)
    (q) edge[loop right] pic[auto] {acc=2} node[below] {$c$} (q);
\end{tikzpicture}}
  \end{subfigure}
  \begin{subfigure}[b]{0.88\linewidth}\centering
	\scalebox{0.8}{
\begin{tikzpicture}[automaton,
  node distance=20mm and 15mm,
  scale=0.6, every node/.style={scale=0.6}]

  \tikzstyle{macrostate}=[state, line width=0.4pt, align=center,
    rectangle, minimum width=42mm, rounded corners=1mm,
    rectangle split, rectangle split parts=3, rectangle split horizontal,
    rectangle split part align={center,center,center},
    rectangle split part fill={blue!10, green!10, orange!10},
    inner sep=2pt]

  \node[macrostate, initial] (s0)
    {$\{0\}$
     \nodepart{two} $\{0\},\{0\}$
     \nodepart{three}
       $C{=}\emptyset$;\;$\macrostateand{\emptyset, \emptyset}{\emptyset}$};

  \node[macrostate, right=of s0] (s2)
    {$\{0,1\}$
     \nodepart{two} $\{0\},\{0\}$
     \nodepart{three}
       $C{=}\{1\}$;\;$\macrostateand{\emptyset, \emptyset}{\emptyset}$};

  \node[macrostate, right=of s2] (s4)
    {$\{0,1\}$
     \nodepart{two} $\{0\},\{0\}$
     \nodepart{three}
       $C{=}\emptyset$;\;$\macrostateand{\emptyset, \emptyset}{\emptyset}$};

  \node[macrostate, below=of s0] (s3)
    {$\{1\}$
     \nodepart{two} $\emptyset,\emptyset$
     \nodepart{three}
       $C{=}\{1\}$;\;$\macrostateand{\emptyset, \emptyset}{\emptyset}$};

  \node[macrostate, right=of s3] (s5)
    {$\{1\}$
     \nodepart{two} $\emptyset,\emptyset$
     \nodepart{three}
       $C{=}\emptyset$;\;$\macrostateand{\{1\}, \{1\}}{\{1\}}$};

  \node[macrostate, right=of s5] (s6)
    {$\{1\}$
     \nodepart{two} $\emptyset,\emptyset$
     \nodepart{three}
       $C{=}\emptyset$;\;$\macrostateand{\{1\}, \emptyset}{\{1\}}$};

  \node[macrostate, right=of s6] (s7)
    {$\{1\}$
     \nodepart{two} $\emptyset,\emptyset$
     \nodepart{three}
       $C{=}\{1\}$;\;$\macrostateand{\{1\}, \emptyset}{\{1\}}$};

  \node[state, below left=10mm and -6mm of s2] (sink) {$\mathit{sink}$};

  \path[->]
    (s0) edge[bend right=30] node[left] {$b,c$} (sink)
    (s0) edge node[above] {$a$} pic[auto]{acc=2} (s2)

    (s2) edge node[above] {$a$} (s4)
    (s2) edge[bend right=10] node[right,pos=0.4] {$b,c$} pic[auto]{acc=1} (s3)
    (s2) edge[bend left=20] node[above,pos=0.4] {$b$} pic[auto]{acc=1} (s5)

    (s4) edge[bend left=10] node[below] {$a$} pic[auto]{acc=2} (s2)
    (s4) edge[bend left=10] node[above right,pos=0.35] {$b,c$} pic[auto]{acc=1} pic[auto,pos=0.62]{acc=2} (s3)

    (s3) edge node[left] {$a$} (sink)
    (s3) edge[loop above] node[above] {$b,c$} pic[auto]{acc=1} (s3)
    (s3) edge node[above] {$b$} pic[auto]{acc=1} (s5)

    (s5) edge[bend left=15] node[above left,pos=0.45] {$a$} (sink)
    (s5) edge node[above] {$b$} pic[auto]{acc=1} (s6)

    (s6) edge[bend right=30] node[below,pos=0.45] {$a$} (sink)
    (s6) edge node[above] {$b$} pic[auto]{acc=1} pic[auto,pos=0.6]{acc=2} (s7)

    (s7) edge[bend right=20] node[right,pos=0.35] {$a$} (sink)
    (s7) edge[bend right=20] node[below] {$b$} pic[auto]{acc=1} (s5)
    (s7) edge[loop right] node[right] {$b$} pic[auto]{acc=1} (s7)

    (sink) edge[loop above] node[above] {$a,b,c$} pic[auto]{acc=0} (sink);

\end{tikzpicture}}
  \end{subfigure}

	\caption{Complementation of the input automaton (left), whose accepting condition is $\accfinof{\tacc{1}} \vee \accinfof{\tacc{2}}$, via the modular construction. The resulting automaton (right) has accepting condition $\accinfof{\tacc{0}} \lor (\accinfof{\tacc{1}} \land \accinfof{\tacc{2}})$.}
  \label{fig:modular-example}
\end{figure}

\section{Elevatorization of ELAs}\label{app:elevatorization}

Given an ELA $\aut$, we detail here how to construct an equivalent elevator ELA $\aut'$.
The procedure first identify the SCCs that violate the elevator property, and then apply an independent semi-determinization algorithm for each of them.
Since semi-determinization produces inherently weak and deterministic components, the resulting automaton will be an elevator ELA.
Let $\aut = (\states, \trans, \inits, \colourset, \colouring, \acccond)$, and let $C_1,\ldots,C_{\ell}$ such that $C_i \subseteq \states$ be the nondeterministic accepting components (NACs) of $\aut$.
For all $i$, we define $\colourset_i$ as the set of colours appearing on the transitions of $C_i$.

Let $\acccond^{\text{DNF}} = (\psi_1 \land \varphi_1) \lor \cdots \lor (\psi_{m} \land \varphi_{m})$ be the disjunctive normal form of $\acccond$ where, for all $j$, $\psi_j$ contains only $\Inf$ predicates and $\varphi_j$ contains only $\Fin$ predicates.
For each $i$, we define $\acccond^{\text{DNF}}_i$ as the simplification of $\acccond^{\text{DNF}}$ for the component $C_i$ such that, for all $c \notin \colourset_i$, $\Finof{c}$ is true and $\Infof{c}$ is false.
We get
\[
\acccond^{\text{DNF}}_i = (\psi_{i, 1} \land \varphi_{i, 1}) \lor \cdots \lor (\psi_{i, m_i} \land \varphi_{i, m_i})
\]
where $\psi_{i, j}$ is of the form $\bigwedge_{k=1}^{n_{i, j}} \Infof{c_{i, j, k}}$ and $\varphi_{i, j}$ is of the form $\bigwedge_{k=1}^{r_{i, j}} \Finof{d_{i, j, k}}$.
Let $\colourset_{i, j}^{\psi} \subseteq \colourset_i$ and $\colourset_{i, j}^{\varphi} \subseteq \colourset_i$ be respectively the set of colours appearing in  $\psi_{i, j}$ and $\varphi_{i, j}$.

The elevator automaton $\aut' = (\states', \trans', \inits', \colourset', \colouring', \acccond')$ is defined as follows.
\begin{itemize}
	\item We have $\states' = \states \cup \bigcup_{i=1}^{\ell} \bigcup_{j=1}^{m_i} D_{i,j}$
	where, for each $C_i$ and each disjunct $(\psi_{i,j} \land \varphi_{i, j})$ of $\acccond^{\text{DNF}}$, the set $D_{i,j}$ is defined as
	$D_{i,j} = \{(i, j, R,B,k) \mid R,B \subseteq C_i,\; B \subseteq R,\; k \in \{0,\ldots,n_{i,j}\}\}$.

	\item The initial states of $\aut'$ remain as in $\aut$, i.e., $\inits' = \inits$.

	\item The set of colours is $\colourset' = \colourset \cup \{\bot, \top\} $, where $\bot, \top \notin \colourset$.
	The fresh colour $\bot$ is used to disable acceptance in the components of $\aut$ that violate the elevator property.
	The fresh colour $\top$, emitted exclusively in $\bigcup_{j=1}^{m_i} D_{i, j}$, is used to ensure the satisfiability of $\psi_{i,j} \land \varphi_{i, j}$.

	\item The acceptance condition is defined by $\acccond' = (\acccond \land \Finof{\bot}) \lor \Infof{\top}$.
	The components of $\aut$ that do not violate the elevator property are not affected.
	However, all $C_i$ in $\aut'$ are labeled with $\bot$, thus enforcing the acceptance to come from visiting $\top$ infinitely in some $D_{i,j}$.

	\item The transition relation and its colouring are defined together.
	\begin{description}
		
		\item[Transitions not embedded in some $\mathbf{C_i}$.]
		If $(q,a,q')\in\trans$ and there is no $i \in \{1, \ldots, \ell\}$ such that $q, q' \in C_i$, then $(q,a,q')\in\trans'$
		and $\colouring'((q,a,q')) = \colouring((q,a,q'))$.
		
		\item[Transitions inside some $\mathbf{C_i}$.]
		If $(q,a,q')\in\trans$ with $q,q'\in C_i$, then $(q,a,q')\in\trans'$ and $\colouring'((q,a,q'))=\{\bot\}$.
		
		\item[Transitions from some $\mathbf{C_i}$ to some $\mathbf{D_{i,j}}$.]
		For all $C_i$ and all $D_{i, j}$, if $(q,a,q')\in\trans$ with $q, q'\in C_i$, then $(q,a,(i, j, \{q'\},\emptyset,0))\in\trans'$ and $\colouring'=\emptyset$.
		
		\item[Transitions in some $\mathbf{D_{i,j}}$.]
		Consider any $C_i$ and any $D_{i, j}$.
		Given a set of states $S \subseteq C_i$, a letter $a \in \Sigma$, and a set of colours $\Lambda \subseteq \colourset'$, we define the set of states obtained by firing transition over $a$ from $S$ within $C_i$ and colored by some subset of $\Lambda$.
		\[
		\delta_{i, j}(S, a, \Lambda) = \{q' \mid q \in S, (q, a, q') \in \trans, q' \in C_i, \colouring((q, a, q')) \subseteq \Lambda\}
		\]
		For $(i, j, R, B, k) \in D_{i,j}$ and $a\in\Sigma$, we define the $a$-successors of $R$ that do not visit any color of $\varphi_{i, j}$ by $R' = \delta_{i, j}(R, a, \colourset \setminus \colourset_{i, j}^{\varphi})$.
		The set $B'_{\text{tmp}} = \delta_{i, j}(B, a, \colourset \setminus \colourset_{i, j}^{\varphi}) \cup \delta_{i, j}(R, a, \{k\} \setminus \colourset_{i, j}^{\varphi})$ is define similarly from $B$ with additionally the states reach from $R$ while visiting the colour $k$ of $\psi_{i, j}$.	
		If $R'\neq\emptyset$, we construct the transition $(i, j, R, B, k) \ltr{a} (i, j, R', B', k')$ where $B'$ and $k'$ are define by
		\[
		B' = \begin{cases}
			\emptyset & \text{if $B'_{\text{tmp}}=R'$},\\
			B'_{\text{tmp}} & \text{otherwise}
		\end{cases}\qquad
		k' =
		\begin{cases}
			(k+1) \bmod [n_{i,j}] & \text{if $B'_{\text{tmp}}=R'$},\\
			k & \text{otherwise}
		\end{cases}
		\]
		and the colouring $\colouring'$ is such that 
		\[
		\colouring'\bigl((i, j, R, B, k), a, (i, j, R',B',k')\bigr) =
		\begin{cases}
			\{\bot, \top\} & \text{if } B'_{\mathrm{tmp}}=R' \text{ and } \ell'=0,\\
			\{\bot\} & \text{otherwise.}
		\end{cases}
		\]
	\end{description}
\end{itemize}

\fi

\end{document}